\documentclass[format=acmsmall, review=false]{acmart}
\usepackage{acm-ec-20}
\usepackage{booktabs} % For formal tables
\usepackage{algorithm}
\usepackage{algpseudocode}
\algnotext{EndFor}
\algnotext{EndIf}
\algnotext{EndWhile}
\algnotext{EndProcedure}

\newcommand{\cA}{\mathcal{A}}
\newcommand{\cD}{\mathcal{D}}
\newcommand{\cG}{\mathcal{G}}
\newcommand{\E}{\mathbb{E}}
\newcommand{\R}{\mathbb{R}}
\newcommand{\talpha}{\tilde{\alpha}}
\newcommand{\tbeta}{\tilde{\beta}}
\newcommand{\bone}{\mathbf{1}}
\renewcommand{\epsilon}{\varepsilon}

\DeclareMathOperator*{\argmax}{argmax}

\newcommand*{\subproofname}{Proof}
\newenvironment{subproof}[1][\subproofname]{\begin{proof}[#1]}{\end{proof}}

\newtheorem{theorem}{Theorem}[section]

\newtheorem{claim}[theorem]{Claim}

\setcitestyle{authoryear}
\allowdisplaybreaks

\title[Closing Gaps in Asymptotic Fair Division]{Closing Gaps in Asymptotic Fair Division}

\author{Pasin Manurangsi}
\affiliation{%
  \institution{Google Research}
  \country{USA}}
\author{Warut Suksompong}
\affiliation{%
  \institution{University of Oxford}
  \country{UK}}

\begin{abstract}
We study a resource allocation setting where $m$ discrete items are to be divided among $n$ agents with additive utilities, and the agents' utilities for individual items are drawn at random from a probability distribution.
Since common fairness notions like envy-freeness and proportionality cannot always be satisfied in this setting, an important question is \emph{when} allocations satisfying these notions exist.
In this paper, we close several gaps in the line of work on asymptotic fair division.
First, we prove that the classical round-robin algorithm is likely to produce an envy-free allocation provided that $m=\Omega(n\log n/\log\log n)$, matching the lower bound from prior work.
We then show that a proportional allocation exists with high probability as long as $m\geq n$, while an allocation satisfying envy-freeness up to any item (EFX) is likely to be present for any relation between $m$ and $n$.
Finally, we consider a related setting where each agent is assigned exactly one item and the remaining items are left unassigned, and show that the transition from non-existence to existence with respect to envy-free assignments occurs at $m=en$.

\end{abstract}

\begin{document}

\maketitle

% Paper body
\section{Introduction}

One of the most frequently occurring tasks in our society is that of allocating scarce resources among interested agents or entities.
Indeed, whether it be apportioning government funds among public organizations, allotting office space to research groups in a university, or assigning houses to residents of a city, one is faced with the decision of how to best allocate the limited resource.
A central concern when making such decisions is \emph{fairness}: it is desirable that all agents view the share they receive as fair.

Among the plethora of fairness notions that have been proposed in the literature, perhaps the two best-known ones are \emph{envy-freeness} and \emph{proportionality}.
An allocation is said to be \emph{envy-free} if it does not induce envy between any pair of agents, and \emph{proportional} if it gives every agent at least $1/n$ of their value for the entire resource, where $n$ denotes the number of agents among whom the resource is divided.
When the resource to be allocated is \emph{divisible}, such as advertisement space or broadcast time, it is known that envy-free and proportional allocations are guaranteed to exist \citep{DubinsSp61,Stromquist80}.
However, in many situations we need to allocate \emph{indivisible} resources like houses, cars, electronics, and musical instruments.
For such discrete items, neither envy-freeness nor proportionality can always be satisfied; this can be most easily seen when two quarrelling siblings try to divide a single toy between themselves.

In light of this negative observation, an important question is \emph{when} allocations satisfying these and other fairness notions are likely to exist.
\citet{DickersonGoKa14} were the first to address this question: under the assumption that utilities are additive and each agent's utilities for individual items are drawn from probability distributions, they established that an envy-free allocation is present with high probability when the number of items $m$ is at least $\Omega(n\log n)$, but not when $m = n+o(n)$.\footnote{\label{footnote:m<n}When $m<n$, any allocation necessarily leaves some agent empty-handed, so no allocation can be envy-free or proportional provided that each agent has a positive utility for every item.}
\citet{ManurangsiSu19} further refined these bounds by exhibiting that existence is in fact likely as long as $m\geq 2n$ if $m$ is divisible by $n$, but unlikely even when $m=\Theta(n\log n/\log\log n)$ if $m$ is not ``almost divisible''\footnote{This means that $m - \ell$ is not divisible by $n$ for all $\ell \in \{-n^{o(1)}, \dots, n^{o(1)}\}$.} by $n$. 
\citet{Suksompong16} investigated the asymptotic existence of proportional allocations---which are weaker than envy-free allocations under the additivity assumption---and showed that such allocations occur with high probability provided that either $m$ is a multiple of $n$ or $m=\omega(n)$.

In this paper, we present several new results on asymptotic fair division and in the process close a number of gaps left open by previous work.
We assume that agents are endowed with additive utilities, and the utility of each agent for each item is drawn independently from a continuous distribution $\cD$ supported on $[0,1]$ whose probability density function is bounded above and below.
We say that an event happens ``with high probability'' if the probability that it happens converges to $1$ as $n$ goes to infinity.

\begin{itemize}
\item First (Section~\ref{sec:EF}), we show that when $m=\Omega(n\log n/\log\log n)$, the round-robin algorithm, which lets the agents take turns picking their favorite item from the remaining items, outputs an envy-free allocation with high probability.
This improves upon the aforementioned upper bound of $m=\Omega(n\log n)$ and, perhaps more importantly, matches the non-existence result in the case that $m$ is not ``almost divisible'' by $n$.
Hence, except for the case where $m$ is ``almost divisible'' by $n$, our result essentially resolves the question of when envy-free allocations exist.
Furthermore, our result gives a separation between the round-robin allocation and the welfare-maximizing allocation: while the latter is likely to be envy-free when $m=\Omega(n\log n)$ \citep{DickersonGoKa14}, it is unlikely to be even proportional, let alone envy-free, when $m = o(n \log n)$ \citep{ManurangsiSu19}.

\item Second (Section~\ref{sec:prop}), we show that if the distribution $\cD$ has mean at most $1/2$,\footnote{We comment on the necessity of this condition in Section~\ref{sec:prop}.} a proportional allocation exists with high probability as long as $m\geq n$; this completely closes the gap for propotionality with respect to such distributions (cf. footnote~\ref{footnote:m<n}).
The result for $m\geq 2n$ is obtained by using the round-robin algorithm and generalizes a prior result of \citet{AmanatidisMaNi17}, which holds only when $\cD$ is the uniform distribution on $[0,1]$.
On the other hand, the case $n\leq m\leq 2n$ is handled using a matching-based algorithm inspired by previous work.

\item Third (Section~\ref{sec:EFX}), we consider \emph{envy-freeness up to any item (EFX)}: an allocation satisfies this property if any envy that an agent has towards another agent can be eliminated by removing \emph{any} item from the latter agent's bundle \citep{CaragiannisKuMo16}.
While it is currently an important open problem whether an EFX allocation always exists, we show that such allocations are likely to exist for any relation between $m$ and $n$.
This complements recent lines of work which show the (non-asymptotic) existence of \emph{approximate} EFX allocations \citep{PlautRo18,AmanatidisNtMa20} and exact EFX allocations \emph{when items can be discarded} \cite{CaragiannisGrHu19,ChaudhuryKaMe20}.

\item Fourth (Section~\ref{sec:assignments}), we analyze the related but slightly different setting of \emph{assignments}, also known as \emph{house allocation}, where each agent is assigned exactly one item and the remaining items are left unassigned.
In this setting, \citet{GanSuVo19} proved that an envy-free assignment is present with high probability if $m=\Omega(n\log n)$, and left open the question of where the transition between non-existence and existence occurs.
We essentially settle this question by showing that this transition occurs at $m=en$: for any constant $\varepsilon > 0$, an envy-free allocation is likely to exist if $m/n \geq e + \varepsilon$, and unlikely to exist if $m/n\leq e - \varepsilon$.
\end{itemize}

Besides closing the gaps themselves, our results also reveal qualitative insights on the relative fairness guarantees provided by different algorithms.
For example, the classical round-robin algorithm performs optimally with respect to envy-freeness, whereas a matching-based algorithm is better suited for proportionality when the number of items is small.

\subsection{Further Related Work}

Fair division is a fascinating topic whose formal study stretches back over half a century \citep{Steinhaus48}; see the books by \citet{BramsTa96} and \citet{Moulin03} for an overview of its long and intriguing history.
While early work in the subject focused on allocating divisible resources (a problem often referred to as \emph{cake cutting}), the fair allocation of indivisible resources has attracted substantial interest from different research communities in the last few years \citep{Thomson16,Markakis17,Moulin19}. 
After the work of \citet{DickersonGoKa14} that initiated the study of asymptotic fair division, \citet{KurokawaPrWa16} and \citet{FarhadiGhHa19} established the probabilistic existence of allocations satisfying a weakening of proportionality called \emph{maximin share fairness}, the latter work also allowing agents to have unequal entitlements.  \citet{ManurangsiSu17} extended Dickerson et al.'s results on envy-freeness to a more general setting where items are allocated to groups instead of to individual agents.

In addition to EFX, another (weaker) relaxation of envy-freeness that has been extensively studied is \emph{envy-freeness up to one item (EF1)}, which requires that any envy that an agent has towards another agent can be eliminated by removing \emph{some} item from the latter agent's bundle \citep{LiptonMaMo04,Budish11}.
Unlike EFX, whose guaranteed existence remains an open question, EF1 can be easily attained for additive utilities using the round-robin algorithm.
%Moreover, it is possible to fulfill EF1 under various restrictions \citep{BarmanKrVa18,BiswasBa18,BiloCaFl19,OhPrSu19}
%as well as in the group setting \citep{KyropoulouSuVo19,SegalhaleviSu19}.

\section{Preliminaries}
\label{sec:prelims}

In our model, a set $M=[m]$ of indivisible items is to be allocated to a set $N=[n]$ of agents, where $[k]:=\{1,2,\dots,k\}$ for any positive integer $k$.
Each agent $i\in N$ has a utility $u_i(j)\geq 0$ for each item $j\in M$. We assume without loss of generality that $u_i(j)\in[0,1]$ for all $i,j$, since otherwise we can simply scale down the utilities by their maximum.
The utilities are \emph{additive}, meaning that $u_i(M') = \sum_{j\in M'}u_i(j)$ for all $M'\subseteq M$.
Additivity is a common assumption in fair division; in particular, to the best of our knowledge, it is assumed in all of the works on asymptotic fair division thus far. 

A \emph{bundle} refers to any subset $M'\subseteq M$ of items.
An \emph{allocation} is a partition of the items into $n$ bundles $(M_1,\dots,M_n)$, where agent $i$ receives bundle $M_i$.
An allocation is said to be \emph{envy-free} if $u_i(M_i)\geq u_i(M_{i'})$ for all $i,i'\in N$, and \emph{envy-free up to any item (EFX)} if $u_i(M_i)\geq u_i(M_{i'}\backslash\{j\})$ for all $i,i'\in N$ and all $j\in M_{i'}$.
For the sake of convenience, when the allocation under consideration is clear, we say that agent $i$ is envy-free (resp., EFX) with respect to agent $i'$ if the corresponding inequality is satisfied for $i$ and $i'$, and that agent $i$ is envy-free (resp., EFX) if the corresponding inequality is satisfied for $i$ and all $i'\in N$.
An allocation is said to be \emph{proportional} if $u_i(M_i)\geq u_i(M)/n$ for all $i\in N$.\footnote{When discussing proportionality, we will sometimes allow allocations to be \emph{partial}, i.e., leave some items unallocated. 
It is clear that if a partial allocation is proportional, then by allocating the remaining items arbitrarily, we obtain a complete proportional allocation.}

For agents $i\in N$ and items $j\in M$, the utilities $u_i(j)$ are drawn independently from a given distribution $\cD$ supported on $[0,1]$.
A distribution is said to be \emph{non-atomic} if it does not put positive probability on any single point.
For a non-atomic distribution $\cD$, we denote by $F_\cD$ and $f_{\cD}$ the cumulative distribution function (CDF) and the probability density function (PDF) of $\cD$ respectively. 
Throughout this work, the assumption that we place on the distributions we consider is that their PDFs are bounded, as stated more precisely below.

\begin{definition} \label{def:pdf-bounded}
For $\alpha,\beta > 0$, we say that a distribution $\cD$ supported on $[0, 1]$ is \emph{$(\alpha, \beta)$-PDF-bounded} if $\cD$ is non-atomic and $\alpha \leq f_{\cD}(x) \leq \beta$ for all $x \in [0, 1]$. 
We say that $\cD$ is \emph{PDF-bounded} if it is $(\alpha, \beta)$-PDF-bounded for some $\alpha, \beta > 0$.
\end{definition}

It follows from the definition that any $(\alpha,\beta)$-PDF-bounded distribution must have $\alpha\leq 1$ and $\beta\geq 1$.
Note that many common distributions, including the uniform distribution on $[0,1]$ (henceforth denoted by $U[0,1]$) and a normal distribution (with any mean and variance) truncated at $0$ and $1$ are PDF-bounded.
For the sake of convenience, we may use the notations $F_X$ and $f_X$ for a random variable $X$ to refer to the CDF and PDF of its associated distribution. Similarly, we say that $X$ is $(\alpha, \beta)$-PDF-bounded if its distribution is $(\alpha, \beta)$-PDF-bounded.
In this paper, we think of $\cD$ as a fixed distribution that does not change with $n$ and $m$; specifically, the parameters $\alpha, \beta$ are constants and our big-O notation may include terms that depend on these parameters.
We say that an event happens \emph{with high probability} if the probability that it happens approaches $1$ as $n\rightarrow\infty$.
Furthermore, when we write $\log n$, the logarithm is assumed to have base $2$.

In the \emph{round-robin algorithm}, the agents take turns picking their favorite item from the remaining items. 
We assume without loss of generality that the order in which the agents pick the items is $1,2,\dots,n,1,2\dots$ until the items run out.
The $t$-th ``round'' consists of each agent's $t$-th pick (so in the last round, not every agent may get to pick).

\subsection{Range-Conditioned Distributions}

In the analysis of the round-robin algorithm, we will often find ourselves dealing with distributions restricted to some range. We provide some useful notation and facts for such distributions next.

For any distribution $\cD$ and any real number $c\in (0,1]$, we use $\cD_{\leq c}$ to denote the conditional distribution of $\cD$ on $[0, c]$ (provided that $F_{\cD}(c) > 0$). Notice that the CDF and PDF of this distribution are
\begin{align} \label{eq:range-conditioned-CDF}
F_{\cD_{\leq c}}(x) = 
\begin{cases}
\frac{F_{\cD}(x)}{F_{\cD}(c)} &\text{ if } x \leq c; \\
1 &\text{ otherwise.}
\end{cases}
\end{align}
and
\begin{align*}
f_{\cD_{\leq c}}(x) =
\begin{cases}
\frac{f_{\cD}(x)}{F_{\cD}(c)} &\text{ if } x \leq c; \\
0 &\text{ otherwise.}
\end{cases}
\end{align*}

Now, let $Y$ be a random variable generated as follows: we draw $X$ from $\cD_{\leq c}$ and set $Y = X/c$. Then, from the above expression of $f_{\cD_{\leq c}}$, we have
\begin{align*}
f_{Y}(y) =
\begin{cases}
\frac{c \cdot f_{\cD}(cy)}{F_{\cD}(c)} &\text{ if } y \leq 1; \\
0 &\text{ otherwise.}
\end{cases}
\end{align*}
The following proposition follows almost immediately from the above expression. (Recall from Definition~\ref{def:pdf-bounded} that a PDF-bounded distribution is implicitly assumed to be supported on $[0, 1]$.)

\begin{proposition} \label{prop:conditional-bounded}
For any $(\alpha, \beta)$-PDF-bounded distribution $\cD$ and any $c \in (0, 1]$, suppose that $Y$ is a random variable where we draw $X$ from $\cD_{\leq c}$ and let $Y = X/c$. Then, $Y$ is $\left(\alpha/\beta, \beta/\alpha\right)$-PDF-bounded.
\end{proposition}

\begin{proof}
Since $\cD$ is $(\alpha, \beta)$-PDF-bounded, we have $\alpha\leq f_\cD(x) \leq \beta$ for all $x \in [0, c]$. This implies that $F_{\cD}(c) \in [c \cdot \alpha, c \cdot \beta]$. Hence, for all $y \in [0, 1]$, we have
\begin{align*}
f_Y(y) = \frac{c \cdot f_{\cD}(cy)}{F_{\cD}(c)} \leq \frac{c \cdot \beta}{c \cdot \alpha} = \frac{\beta}{\alpha}
\end{align*}
and
\begin{align*}
f_Y(y) = \frac{c \cdot f_{\cD}(cy)}{F_{\cD}(c)} \geq \frac{c \cdot \alpha}{c \cdot \beta} = \frac{\alpha}{\beta}.
\end{align*}
Hence $Y$ is $\left(\alpha/\beta, \beta/\alpha\right)$-PDF-bounded, as desired.
\end{proof}

For any real number $c < 1$ (such that $F_{\cD}(c) < 1$), we also define $\cD_{> c}$ in a similar manner as $\cD_{\leq c}$ above.

\subsection{(Anti-)Concentration Inequalities}

We will need a few (anti-)concentration inequalities for sums of independent random variables. Our first inequality is the standard Chernoff bound:

\begin{lemma}[Chernoff bound] \label{lem:chernoff}
Let $X_1, \dots, X_k$ be independent random variables taking values in $[0, 1]$, and let $S := X_1 + \cdots + X_k$. Then, for any $\delta \geq 0$,
$$\Pr[S \geq (1 + \delta)\E[S]] \leq \exp\left(\frac{-\delta^2 \E[S]}{3}\right)$$
and
$$\Pr[S \leq (1 - \delta)\E[S]] \leq \exp\left(\frac{-\delta^2 \E[S]}{2}\right).$$
\end{lemma}

Our next inequality is a sharpening of (the first case of) the above Chernoff bound for a certain regime. The formal statement of the bound is stated below; we note here that the important case is when we would like to bound $\Pr[Y_1 + \cdots + Y_r \geq r - O(1)]$. In this case, the Chernoff bound\footnote{In particular, we may apply the first inequality in Lemma~\ref{lem:chernoff} with $k = r$ (so $\E[S]$ is linear in $r$) and $\delta = O(1)$.} gives an $n^{-\Omega(1)}$ bound only when $r = \Omega(\log n)$, whereas the following inequality gives an $n^{-\Omega(1)}$ bound even when $r = \Omega(\log n / \log \log n)$ (by letting $c = O(1)$ and $d = \frac{r}{2c}$). This $\log \log n$ saving in $r$ is precisely what will allow us to establish the (asymptotically) tight bound for the existence of envy-free allocations, which we do in Section~\ref{sec:EF}.

\begin{lemma} \label{lem:concen-sum}
Let $r, c, d$ be any positive integers such that $r \geq cd$, and let $Y_1, \dots, Y_r$ be independent random variables that are $(\talpha, \tbeta)$-PDF-bounded. Then, we have
\begin{align*}
\Pr[Y_1 + \cdots + Y_r \geq r - c] \leq 2^r \left(\frac{\tbeta}{d}\right)^{r - cd}.
\end{align*} 
\end{lemma}

\begin{proof}
The inequality is trivial if $\tbeta > d$, so we may assume that $\tbeta \leq d$.
Let $S$ denote the set of indices $i$ such that $Y_i \leq 1 - \frac{1}{d}$. If $Y_1 + \cdots + Y_r \geq r - c$, we must have $|S| \leq cd$. As a result, by union bound, we have
\begin{align*}
\Pr[Y_1 + \cdots + Y_r \geq r - c] 
&\leq \sum_{S \subseteq [r], |S| \leq cd} \Pr\left[\forall i \notin S, Y_i > 1 - \frac{1}{d}\right] \\
&\leq \sum_{S \subseteq [r], |S| \leq cd} \left(\frac{\tbeta}{d}\right)^{r - |S|} \\
&\leq 2^r \left(\frac{\tbeta}{d}\right)^{r - cd},
\end{align*}
as claimed.
\end{proof}

Our last inequality is of the opposite nature from the above bounds: it says that the probability that $X_1 + \cdots + X_k$ is ``far'' from its expectation is still large (i.e., anti-concentration). 

\begin{lemma} \label{lem:anti-concen-sum}
Let $X_1, \dots, X_k$ be independent random variables sampled from $\cD$ whose support is a subset of $[0, 1]$, and let $S := X_1 + \cdots + X_k$. 
Suppose that $\cD$ has variance $\sigma^2 > 0$ and mean at most $1/2$.
Then, $$\Pr\left[S \leq \frac{k}{2} - 0.01\sigma \sqrt{k}\right] > 0.49 - O\left(\frac{1}{\sqrt{k}}\right),$$
where the constant in the big-O notation can depend on $\sigma$.
\end{lemma}

\begin{proof}
Let $\mu \leq 1/2$ denote the mean of $\cD$.
Consider $S' := \frac{1}{\sigma \sqrt{k}}(S - k \cdot \mu)$. From the Berry-Esseen theorem~\cite{Berry41,Esseen42}, the CDF of $S'$ and that of the standard normal distribution $\Phi$ differ by at most $O(1/\sqrt{k})$ pointwise. It follows that
\begin{align*}
\Pr\left[S \leq \frac{k}{2} - 0.01\sigma\sqrt{k}\right] 
&\geq \Pr\left[S \leq k \cdot \mu - 0.01\sigma\sqrt{k}\right] \\
&= \Pr[S' \leq -0.01] \\
&\geq \Phi(-0.01) - O\left(\frac{1}{\sqrt{k}}\right) \\
&> 0.49 - O\left(\frac{1}{\sqrt{k}}\right). \qedhere
\end{align*}
\end{proof}

\subsection{An Inequality for the Round-Robin Algorithm}

We now present a lemma related to the round-robin algorithm. To state this lemma, let us introduce another notation: for any distribution $\cD$ and any positive integer $k$, we use $\cD^{\max(k)}$ to denote the distribution of the maximum of $k$ independent random variables distributed according to $\cD$. 

% In the round-robin algorithm, it is simple to see that, if we consider an agent $i$ and let $X_r$ be his value for the item that he gets in round $r$ (and $X_0 = 1$ for convenience), then $X_r$ is drawn from $\cD^{\max(m + 1 - i - n(r - 1))}_{\leq X_{r - 1}}$. We often want to show that $X_r$ is large for a specified $r$. We make a generic calculation below, which will be used multiple times in this work.

In the round-robin algorithm, consider any agent $i$ and let $X_r$ be his value for the item that he gets in round $r$ (and $X_0 = 1$ for convenience). We will show later (in Lemma~\ref{lem:rr-simplified-process}) that $X_1,X_2, \dots$ are distributed as if they were drawn from the following process: for $r = 1, 2,\dots$, sample $X_r$ from $\cD^{\max(m + 1 - i - n(r - 1))}_{\leq X_{r - 1}}$. We often want to show that $X_r$ is large for a specified $r$. We make a generic calculation below, which will be used multiple times in this work.

\begin{lemma} \label{lem:rep-sel}
Let $T$ be a positive integer and let $s_1, \dots, s_T$ be any positive integers. Consider the random variables $X_0 = 1$ and $X_1, \dots, X_T$ generated by the following process: for every $t = 0,1, \dots, T - 1$, sample $X_{t + 1}$ according to $\cD^{\max(s_{t+1})}_{\leq X_t}$. If $\cD$ is $(\alpha, \beta)$-PDF-bounded, then for any parameter $p \in (0, 1)$ we have
\begin{align*}
\Pr\left[X_T \geq 1 - \frac{\beta}{\alpha} \cdot \frac{T\ln(T/p)}{s}\right] \geq 1 - p,
\end{align*}
where $s := \min\{s_1, \dots, s_T\}$.
\end{lemma}

\begin{proof}
The inequality trivially holds if the expression $1 - \frac{\beta}{\alpha} \cdot \frac{T\ln(T/p)}{s}$ is negative, so we may assume that this expression is nonnegative, which means that $1 - \frac{\beta}{\alpha} \cdot \frac{\ln(T/p)}{s}$ is also nonnegative.
For every $t = 0,1, \dots, T - 1$, Proposition~\ref{prop:conditional-bounded} implies that if we sample $Z \sim \cD_{\leq X_t}$ and let $Y = Z/X_t$, then $Y$ is $(\alpha/\beta, \beta/\alpha)$-PDF-bounded. As a result, we have
\begin{align*}
\Pr\left[Z < \left(1 - \frac{\beta}{\alpha} \cdot \frac{\ln(T/p)}{s}\right)X_t\right] &= \Pr\left[Y< 1 - \frac{\beta}{\alpha} \cdot \frac{\ln(T/p)}{s}\right] \\ 
(\text{From } (\alpha/\beta, \beta/\alpha)\text{-PDF-boundedness of } Y) &\leq 1 - \frac{\alpha}{\beta} \cdot \frac{\beta}{\alpha} \cdot \frac{\ln(T/p)}{s} \\
&= 1 - \frac{\ln(T/p)}{s}.
\end{align*}

From the above inequality and since $X_{t + 1}$ is sampled from $\cD_{\leq X_t}^{\max(s_{t + 1})}$, we have
\begin{align*}
\Pr\left[X_{t + 1} < \left(1 - \frac{\beta}{\alpha} \cdot \frac{\ln(T/p)}{s}\right) X_t\right] \leq \left(1 - \frac{\ln(T/p)}{s}\right)^{s_{t + 1}} 
\leq e^{-\frac{\ln(T/p)}{s} \cdot s_{t + 1}}
\leq e^{-\ln(T/p)}
= p/T,
\end{align*}
where the second inequality follows from the well-known inequality $1-x\leq e^{-x}$, which holds for any real number $x$.

Hence, by union bound, with probability at least $1 - p$, we have $X_{t + 1} \geq \left(1 - \frac{\beta}{\alpha} \cdot \frac{\ln(T/p)}{s}\right) X_t$ for all $t = 0,1, \dots, T - 1$.
When this is the case, we get
\begin{align*}
X_T \geq \left(1 - \frac{\beta}{\alpha} \cdot \frac{\ln(T/p)}{s}\right)^T X_0 \geq 1 - \frac{\beta}{\alpha} \cdot \frac{T\ln(T/p)}{s},
\end{align*}
where we use Bernoulli's inequality for the second inequality and the assumption that $X_0=1$. This completes the proof of the lemma.
\end{proof}

\section{Envy-freeness}
\label{sec:EF}

In this section, we consider envy-freeness. Our main result is the following theorem:

\begin{theorem} \label{thm:rr-ef}
Suppose that $\cD$ is PDF-bounded.
For $m = \Omega\left(\frac{n \log n}{\log\log n}\right)$, the round-robin algorithm outputs an envy-free allocation with high probability.
\end{theorem}

Since an envy-free allocation is unlikely to exist even when $m=\Theta(n\log n/\log\log n)$ if $m$ is not ``almost divisible'' by $n$ and the constant in the asymptotic notation is sufficiently small \citep{ManurangsiSu19}, the bound in Theorem~\ref{thm:rr-ef} is asymptotically tight.

Let us introduce an additional notation, which we will use in this section as well as in Section~\ref{sec:prop}. 
For any agents $i, i'$, denote by $X^{i, i'}_t$ agent $i$'s utility for the $t$-th item received by agent $i'$ in the round-robin algorithm, and let $X^{i,i'}_0 = 1$ for convenience. When $i = i'$, we abbreviate $X^{i, i'}_t$ as $X^i_t$. The following lemma, which was alluded to before Lemma~\ref{lem:rep-sel}, allows us to consider a simple random process that generates $(X^{i, i'}_t)_{i, i' \in [n], t \in \left[1 + \left\lfloor \frac{m - i'}{n} \right\rfloor\right]}$ rather than dealing with the round-robin algorithm directly. 

\begin{lemma} \label{lem:rr-simplified-process}
$(X^{i, i'}_t)_{i, i' \in [n], t \in \left[1 + \left\lfloor \frac{m - i'}{n} \right\rfloor\right]}$ has the same distribution as if it is generated as follows:
\begin{enumerate}
\item For $t = 1, \dots, \lceil m/n \rceil$:
\begin{enumerate}
\item For $i = 1, \dots, \min\{n, m - (t - 1)n\}$: 
\label{step:rr-item-selection}
\begin{enumerate}
\item Sample $X^i_t \sim \cD_{\leq X^i_{t - 1}}^{\max(m + 1 - (t - 1)n - i)}$. \label{step:rr-item-max-util}
\item For every $1 \leq i' < i$, sample $X^{i', i}_t \sim \cD_{\leq X^{i'}_t}$.
\item For every $i < i' \leq n$, sample $X^{i', i}_t \sim \cD_{\leq X^{i'}_{t - 1}}$.
\end{enumerate}
\end{enumerate}
\end{enumerate}
\end{lemma}

The proof of Lemma~\ref{lem:rr-simplified-process} is deferred to the appendix. We note here that each loop in Step~\ref{step:rr-item-selection} should be thought of as agent $i$ choosing his/her $t$-th item. Observe also that $m + 1 - (t - 1)n - i$ is simply the number of remaining items before agent $i$'s choice is made, and Step~\ref{step:rr-item-max-util} can be thought of as picking the best among these items.

Before we prove Theorem~\ref{thm:rr-ef}, let us give the high-level intuition behind the proof. 
Consider any pair of agents $i,i'$. If $i<i'$, then clearly $i$ does not envy $i'$, so we focus on the case where $i>i'$. For $i$ to envy $i'$, we must have $u_i(M_{i'}) - u_i(M_i) > 0$. Note that $u_i(M_{i'}) - u_i(M_i)$ is at most
\begin{align*}
X^{i, i'}_1 - (X^i_1 - X^{i, i'}_2) - (X^i_2 - X^{i, i'}_3) - \cdots - (X^i_{z-1} - X^{i, i'}_z),
\end{align*}
where $z$ is the last round in which $i'$ picks an item.
In other words, $X^{i, i'}_1$ is the amount that $i'$ ``gains'' with her first item from $i$'s viewpoint, and $(X^i_1 - X^{i, i'}_2), (X^i_2 - X^{i, i'}_3), \dots,(X^i_{z-1} - X^{i, i'}_z)$ are the gains of $i$ that $i$ will use to try to ``catch up'', so that $i$ does not envy $i'$ in the end. Note that the first gain $X^{i, i'}_1$ of $i'$ is rather small, i.e., no more than $1$. Moreover, each of the gains $(X^i_t - X^{i, i'}_{t + 1})$ can be written as $X^i_t(1 - X^{i, i'}_{t + 1}/X^i_t)$. Lemma~\ref{lem:rep-sel} ensures that these ``scaling factors'' $X^i_t$ are relatively large, which allows us to apply the concentration of sums of independent random variables from Lemma~\ref{lem:concen-sum} to $X^{i, i'}_{t + 1}/X^i_t$. The proof below implements this idea with the appropriate selection of parameters. 

\begin{proof}[Proof of Theorem~\ref{thm:rr-ef}]
Suppose that $\cD$ is $(\alpha, \beta)$-PDF-bounded and that $m \geq 10^6 \tbeta \log(10^6 \tbeta) \cdot \frac{n\log(10n)}{\log \log(10n)}$, where $\tbeta := \beta/\alpha$. Consider any pair of distinct agents $i, i'$. We will show that the probability that $i$ envies $i'$ is at most $O(1/m^3)$. The union bound then allows us to immediately arrive at the desired result. 

Observe that, if $i < i'$, then clearly $i$ does not envy $i'$. As a result, we may henceforth assume that $i > i'$. In this case, from Lemma~\ref{lem:rr-simplified-process}, we may view the values $X^{i, i'}_1, X^i_1, X^{i, i'}_2, \dots$ as being generated by the following process:
\begin{enumerate}
\item Randomly sample $X^{i, i'}_1$ from $\cD$.
\item For $t = 1,2, \dots, 1 + \left\lfloor \frac{m - i}{n} \right\rfloor$:
\begin{enumerate}
\item Randomly sample $X^i_t$ from $\cD_{\leq X^i_{t - 1}}^{\max(m + 1 - i - n(t - 1))}$.
\item Randomly sample $X^{i, i'}_{t + 1}$ from $\cD_{\leq X^i_t}$.
\end{enumerate}
\end{enumerate}
Note that in the last iteration it may be possible that $X^{i, i'}_{t+1}$ is not actually selected in the round-robin algorithm if all the items are already assigned. In this case we overestimate $i$'s value for $i'$'s bundle; hence, we will still get an upper bound on the probability that $i$ envies $i'$.

Let us also define $Y^{i, i'}_{t} := X^{i, i'}_{t + 1}/X^i_t$.
Recall that $i$ does not envy $i'$ if $u_i(M_i) \geq u_i(M_{i'})$. To bound the probability that this happens, let us rearrange $u_i(M_i) - u_i(M_{i'})$ as
\begin{align}
u_i(M_i) - u_i(M_{i'}) &= \left(\sum_{t=1}^{1 + \left\lfloor\frac{m - i}{n}\right\rfloor} X^i_t\right) - \left(\sum_{t=1}^{2 + \left\lfloor\frac{m - i}{n}\right\rfloor} X^{i, i'}_t\right) \nonumber \\
&= \left(\sum_{t=1}^{1 + \left\lfloor\frac{m - i}{n}\right\rfloor} (X^i_t - X^{i, i'}_{t + 1})\right) - X^{i, i'}_1 \nonumber \\
&\geq \left(\sum_{t=1}^{1 + \left\lfloor\frac{m - i}{n}\right\rfloor} X^i_t \cdot (1 - Y^{i, i'}_t) \right) - 1. \label{eq:diff-rr}
\end{align}

Let $T := 100 \left\lceil \tbeta \cdot \frac{\log m}{\log \log m} \right\rceil$. 
Recall from Definition~\ref{def:pdf-bounded} that $\tbeta = \beta/\alpha \geq 1$, and notice that 
\begin{align}
m/2 - nT &= n\left(\frac{m}{4n} + \frac{m}{4n} - T\right) \nonumber \\
\left(\text{From } m \geq 10^6 \tbeta \log(10^6 \tbeta) \cdot \frac{n \log(10n)}{\log \log(10n)}\right) &\geq n\left(1000 \tbeta \frac{\log(10n)}{\log \log(10n)} + 1000\tbeta \log(m/n) - T\right) \nonumber \\
(\text{From } 1 \leq \log \log (10n) \leq \log \log m) &\geq n\left(1000 \tbeta \frac{\log(10n)}{\log \log m} + 1000\tbeta \frac{\log(m/n)}{\log \log m} - T\right) \nonumber \\
&= n\left(1000 \tbeta \cdot \frac{\log (10m)}{\log \log m} - T\right) \nonumber \\
&\geq 0, \label{eq:bound-on-index-T}
\end{align}
where in the first inequality we also use the fact that $\frac{x}{\log x}$ is non-decreasing in the range $x \in [e, \infty)$, which implies that $\frac{m/n}{\log(m/n)} \geq \frac{10^6 \tbeta \log(10^6 \tbeta)}{\log(10^6 \tbeta \log(10^6 \tbeta))} \geq 4000 \tbeta$.
This in turn implies that $T \leq \frac{m}{2n} \leq 1 + \left\lfloor\frac{m - i}{n}\right\rfloor$, i.e., that $X^{i, i'}_T, X^i_T$ are well-defined.

We now consider two ``bad'' events:
\begin{enumerate}
\item[E1]: The event that $Y^{i, i'}_1 + \cdots + Y^{i, i'}_T \geq T - 2$.
\item[E2]: The event that $X^i_T < 1/2$.
\end{enumerate}

Let us bound the probability that each event occurs. For the event E1, consider any $x_1^i, \dots, x_T^i \in [0, 1]$. Conditioned on $X_1^i = x_1^i, \dots, X_T^i = x_T^i$, Proposition~\ref{prop:conditional-bounded} implies that $Y_1^{i, i'}, \dots, Y^{i, i'}_T$ are $(\alpha/\beta, \tbeta)$-PDF-bounded, and observe that they are independent. Hence, by applying Lemma~\ref{lem:concen-sum} with $r = T$, $c = 2$, and $d = T/4$, we have
\begin{align*}
\Pr[\text{E1 occurs} \mid X_1^i = x_1^i, \dots, X_T^i = x_T^i] &\leq 2^T \left(\tbeta \cdot \frac{4}{T}\right)^{T/2} \\
&= \left(\tbeta \cdot \frac{16}{T}\right)^{T/2} \\
(\text{From our choice of } T) &\leq \left(\frac{\log \log m}{\log m}\right)^{6 \frac{\log m}{\log \log m}} \\
(\text{Since } \log\log m \leq \sqrt{\log m} \text{ for any } m \geq 10^6) &\leq \left(\frac{1}{\sqrt{\log m}}\right)^{6 \frac{\log m}{\log \log m}}  \\
&= \frac{1}{m^3}.
\end{align*}
Hence, we have
\begin{align*}
&\Pr[\text{E1 occurs}] \\
&= \int_0^1 \cdots \int_0^1 \quad \Pr[\text{E1 occurs} | X_1^i = x_1^i, \dots, X_T^i = x_T^i] \cdot f_{X_1^i, \dots, X_T^i}(x_1^i, \dots, x_T^i) \quad d x_1^i \cdots d x_T^i  \\
&\leq \int_0^1 \cdots \int_0^1 \frac{1}{m^3} \cdot f_{X_1^i, \dots, X_T^i}(x_1^i, \dots, x_T^i) \quad d x_1^i \cdots d x_T^i \\
&= \frac{1}{m^3},
\end{align*}
where $f_{X_1^i, ..., X_T^i}$ denotes the PDF of the joint distribution over $X_1^i, ..., X_T^i$ (which are \emph{not} independent random variables).

As for the event E2, notice that the sequence $X^i_0, X^i_1, \dots, X^i_T$ is sampled in the same way as that in Lemma~\ref{lem:rep-sel} with $s_t = m + 1 - i - n(t - 1)$ for $t=1,\dots,T$. Observe also that, for sufficiently large $m$, we have $\frac{\beta}{\alpha} \cdot \frac{T\ln(T \cdot m^3)}{m/2} \leq 1/2$ and $s_1, \dots, s_T \geq m + 1 - i - n(T - 1) \geq m - nT \geq m/2$, where the last inequality follows from~\eqref{eq:bound-on-index-T}. As a result, by plugging in Lemma~\ref{lem:rep-sel} with $p = 1 / m^3$, we have $\Pr[\text{E2 occurs}] \leq \frac{1}{m^3}$.

Hence, by union bound, the probability that at least one of the two bad events occurs is at most $O(1/m^3)$. Now, when neither E1 nor E2 occurs, we can further bound~\eqref{eq:diff-rr} as
\begin{align*}
u_i(M_i) - u_i(M_{i'}) &\geq \left(\sum_{t=1}^{1 + \left\lfloor\frac{m - i}{n}\right\rfloor} X^i_t \cdot (1 - Y^{i, i'}_t) \right) - 1 \\
&\geq \left(\sum_{t=1}^{T} X^i_t \cdot (1 - Y^{i, i'}_t) \right) - 1 \\
(\text{since E2 does not occur}) &\geq \left(\frac{1}{2} \sum_{t=1}^{T} (1 - Y^{i, i'}_t) \right) - 1 \\
(\text{since E1 does not occur}) &\geq \left(\frac{1}{2} \cdot 2 \right) - 1 = 0,
\end{align*} 
meaning that $i$ does not envy $i'$. As a result, the probability that $i$ does not envy $i'$ is at least $1 - O(1/m^3)$. By taking a union bound over all pairs $i, i'$, we have that the allocation is envy-free with probability at least $1 - O(1/m)$, completing the proof.
\end{proof}

\section{Proportionality}
\label{sec:prop}

In this section, we investigate another fundamental fairness notion, proportionality, and establish the following result:

\begin{theorem}
\label{thm:prop}
Suppose that $\cD$ is PDF-bounded and has mean at most $1/2$.
For any $m \geq n$, there is a polynomial-time algorithm that outputs a proportional allocation with high probability.
\end{theorem}

The assumption that $\cD$ has mean at most $1/2$ is necessary to guarantee the existence of a proportional allocation for all $m \geq n$ with high probability. 
To see this, suppose that $\cD$ has mean $1/2 + \varepsilon$ for some constant $\varepsilon > 0$, and let $m = 2n - 1$. 
In this case, the expected value of $u_1(M),\dots,u_n(M)$ is $n(1+2\varepsilon-o(1))$; by standard Chernoff and union bounds, we have that with high probability, $u_i(M) > n$ simultaneously for all $i$. 
When this happens, any allocation cannot be proportional---indeed, it is not proportional for an agent who receives at most one item (since each item has value at most $1$), and such an agent always exists due to the pigeonhole principle.

The proof of Theorem~\ref{thm:prop} will be divided into two parts according to the range of $m$.
For the case $m\geq 2n$ we will again employ the round-robin algorithm (Theorem~\ref{thm:rr-prop}), while the case $n\leq m\leq 2n$ will be handled using a matching-based algorithm (Theorem~\ref{thm:matching-prop}).

\subsection{The Case $m\geq 2n$}
\label{sec:prop-rr}

We begin by showing that the allocation produced by the round-robin algorithm, in addition to satisfying EF1 with certainty and envy-freeness with high probability, is also likely to be proportional even for a modest number of items. 

\begin{theorem} \label{thm:rr-prop}
Suppose that $\cD$ is PDF-bounded and has mean at most $1/2$.
When $m \geq 2n$, the round-robin algorithm outputs a proportional allocation with high probability. 
\end{theorem}

Theorem~\ref{thm:rr-prop} generalizes a result of \citet{AmanatidisMaNi17}, who showed an analogous existence but only for the uniform distribution $U[0, 1]$. 
We remark here that the requirement $m \geq 2n$ is tight. In particular, suppose that $m = 2n - 1$ and that $\cD = U[0, 1]$. Then there is at least a constant probability that $u_n(M) > n$. When this happens, the round-robin algorithm will surely fail, as it only assigns one item (of value at most 1) to the last agent $n$.

To illustrate the high-level idea of the proof, it is most useful to focus on the case $m = 2n$ and $\cD = U[0, 1]$. In this case, Lemma~\ref{lem:rep-sel} implies that with high probability, every agent receives an item that she values $1 - \tilde{O}(1/n)$ in the first round.\footnote{Here we use the notation $\tilde{O}$ to hide a multiplicative factor $O((\log n)^c)$ for some constant $c$.} For simplicity of the discussion, let us ignore the $\tilde{O}(1/n)$ term and assume that each agent receives an item of value $1$ in the first round. Now, from Chernoff bound, we have that every agent is likely to value the whole set $M$ at most $n(1 + \tilde{O}(1/\sqrt{n}))$. As a result, in order for the allocation to be proportional, it suffices for each agent to receive an item of value at least $\tilde{O}(1/\sqrt{n})$ in the second round. For the last agent $n$, the utility she receives in the second round is, roughly speaking, drawn from $U[0, 1]$; hence, the ``bad'' event for this agent occurs with probability $\tilde{O}(1/\sqrt{n})$. For the second-to-last agent $n - 1$, her item's value in the second round is approximately drawn from $\cD^{\max(2)}$; hence, the bad event for this agent occurs with probability $\tilde{O}(1/n)$. Similarly, for each of the remaining agents, the probability of the bad event occurring is at most $\tilde{O}(1/n^{1.5})$. Taking a union bound over all agents yields the desired result.

The formal proof follows the idea outlined above, with some additional steps to ensure that we can appropriately deal with the more general setup. For instance, in the case $m = 2.5n$, it no longer suffices for agent $n$ to receive value $\tilde{O}(1/\sqrt{n})$ in the second round, because now $u_n(M)/n$ is $1.25 \pm o(1)$ instead of $1\pm o(1)$. However, this is in fact an easier case: Lemma~\ref{lem:rep-sel} already implies that the second item that each agent receives is also of value $1 - o(1)$.

\begin{proof}[Proof of Theorem~\ref{thm:rr-prop}]
When $m = \Omega\left(\frac{n \log n}{\log\log n}\right)$, Theorem~\ref{thm:rr-ef} already implies that the round-robin algorithm produces an envy-free (and therefore proportional) allocation with high probability. Hence, it suffices to prove the statement for the case where $m = O(n\log n)$.
%; in this case, $m \to \infty$ implies $n \to \infty$. Hence, we may bound the probability in terms of $n$ instead of $m$.

Suppose that $\cD$ is an $(\alpha, \beta)$-PDF-bounded distribution with mean at most $1/2$.
One can check using Chernoff and union bounds that with high probability, we have
\begin{align} \label{eq:sum-bound-rr}
u_i(M) \leq \frac{m}{2} + 10\sqrt{m \log n}
\end{align}
for all agents $i \in N$.
We will henceforth assume that~\eqref{eq:sum-bound-rr} holds for all $i \in N$.

Let us write $m$ as $nr + q$, where $r = \lfloor m/n \rfloor \geq 2$ and $r= O(\log n)$. We will consider two cases, depending on whether $q \leq n^{0.1}$. (Note that the threshold for $q$ can be any value that is $\omega(\log^2 n)$ and $o(n)$; we simply select $n^{0.1}$ for concreteness.)
Define the notation $X^i_t$ as in Section~\ref{sec:EF}. 
By Lemma~\ref{lem:rr-simplified-process}, the sequence $X^i_0, X^i_1, \dots, X^i_r$ is sampled in the same way as that in Lemma~\ref{lem:rep-sel} with $s_t = m + 1 - i - n(t - 1) \geq q$ for $t=1,\dots,r$.

\paragraph{Case 1: $q \geq n^{0.1}$.}
Consider an agent $i$.
Substituting $p = 1/n^2$ in Lemma~\ref{lem:rep-sel} implies that the following holds with probability at least $1 - 1/n^2$: 
\begin{align*}
X^i_r \geq 1 - \frac{\beta}{\alpha} \cdot \frac{r\ln(r/p)}{q} = 1 - O\left(\frac{\log^2 n}{n^{0.1}}\right).
\end{align*}
This means that the agent receives an item of value at least $1 - O(\log^2 n/n^{0.1})$ in each of the first $r$ rounds. As a result, $i$'s utility for his bundle is at least $r\left(1 - O(\log^2 n/n^{0.1})\right)$. Furthermore, \eqref{eq:sum-bound-rr} ensures that $u_i(M)/n$ is at most 
\begin{align*}
\frac{m}{2n} + O\left(\frac{\sqrt{m\log n}}{n}\right) \leq \frac{r + 1}{2} + O\left(\frac{\log n}{\sqrt{n}}\right),
\end{align*}
where we use our assumption $m=O(n\log n)$.
Since $r \geq 2$, we have $r > \frac{r + 1}{2}$. It follows that for sufficiently large $n$, the allocation is proportional for $i$. By taking a union bound over all agents, we conclude that the allocation is proportional with high probability.

\paragraph{Case 2: $q < n^{0.1}$.} In this case, we will have to give a more refined bound which differs for each agent $i \in N$, with the agents near the end of the round-robin ordering having a worse probability bound. This bound is stated formally below.

\begin{claim} \label{claim:proport-each-agent}
For each $i \in N$, the allocation output by the round-robin algorithm is proportional for $i$ with probability at least $1 - O(1/n^{\min\{0.05(n - i + 1), 2\}})$.
\end{claim}

\begin{subproof}
Observe that $s_t\geq n$ for $t=1,\dots,r-1$.
Substituting $p = 1/n^2$ in Lemma~\ref{lem:rep-sel} implies that the following holds with probability at least $1 - 1/n^2$: 
\begin{align} \label{eq:penultimate}
X^i_{r - 1} \geq 1 - \frac{\beta}{\alpha} \cdot \frac{r\ln(r/p)}{n} = 1 - O\left(\frac{\log^2 n}{n^{0.1}}\right).
\end{align}
This means that agent $i$ receives an item of value at least $1 - O(\log^2 n/n^{0.1})$ in each of the first $r - 1$ rounds. Hence, from the first $r - 1$ rounds, the agent already has a bundle with utility at least $(r - 1)\left(1 - O(\log^2 n / n^{0.1})\right)$. From~\eqref{eq:sum-bound-rr}, it suffices for the agent to receive the following utility in the $r$-th round for him to consider the allocation to be proportional:
\begin{align*}
&\left(\frac{m}{2n} + \frac{10\sqrt{m\log n}}{n}\right) - (r - 1)\left(1 - O\left(\frac{\log^2 n}{n^{0.1}}\right)\right) \\
(\text{From } m = O(n \log n)) &\leq \left(\frac{nr + q}{2n} + O\left(\frac{\log n}{\sqrt{n}}\right)\right) - (r - 1)\left(1 - O\left(\frac{\log^2 n}{n^{0.1}}\right)\right) \\
(\text{From } r \geq 2) &\leq \frac{q}{2n} + O\left(\frac{\log n}{\sqrt{n}}\right) + (r - 1) \cdot O\left(\frac{\log^2 n}{n^{0.1}}\right) \\
(\text{From } q < n^{0.1} \text{ and } r = O(\log n)) &\leq O\left(\frac{\log^3 n}{n^{0.1}}\right).
\end{align*}
Recall from Lemma~\ref{lem:rr-simplified-process} that the item that agent $i$ receives in the $r$-th round has value distributed as the maximum of $m + 1 - i - n(r - 1) \geq n - i + 1$ i.i.d. random variables from $\cD_{\leq X^i_{r - 1}}$. As a result, conditioned on~\eqref{eq:penultimate} occurring, the probability that the allocation is proportional for $i$ is at least
\begin{align*}
1 - F_{\cD_{\leq X^i_{r - 1}}^{\max(n - i + 1)}}\left(O\left(\frac{\log^3 n}{n^{0.1}}\right)\right) 
&= 1 - \left(F_{\cD_{\leq X^i_{r - 1}}}\left(O\left(\frac{\log^3 n}{n^{0.1}}\right)\right)\right)^{n - i + 1} \\
(\text{From Proposition~\ref{prop:conditional-bounded}}) &\geq 1 - \left(\frac{\beta}{\alpha} \cdot \frac{O\left(\frac{\log^3 n}{n^{0.1}}\right)}{X^i_{r - 1}} \right)^{n - i + 1} \\ 
(\text{Since we condition on~\eqref{eq:penultimate}}) &\geq 1 - \left(\frac{\beta}{\alpha} \cdot \frac{O\left(\frac{\log^3 n}{n^{0.1}}\right)}{1 - O\left(\frac{\log^2 n}{n^{0.1}}\right)}\right)^{n - i + 1} \\
&\geq 1 - O(n)^{-0.05(n - i + 1)}
\\
&\geq 1 - O(n)^{-\min\{0.05(n - i + 1),2\}} 
= 1 - O\left(n^{-\min\{0.05(n - i + 1),2\}}\right).
\end{align*}
From this and the fact that \eqref{eq:penultimate} occurs with probability at least $1 - O(1/n^2)$, we get the desired claim.
\end{subproof}

Thus, by union bound, the allocation produced by the round-robin algorithm fails to be proportional with probability at most $$\sum_{i=1}^{n} O\left(1/n^{\min\{0.05(n - i + 1), 2\}}\right) 
\leq \sum_{i=1}^{n-40} O(1/n^2) + \sum_{i=n-39}^n O(1/n^{0.05})
= O(1/n^{0.05}) = o(1),$$ which concludes Case~2 and therefore our proof.
\end{proof}

\subsection{The Case $n \leq m \leq 2n$}
\label{sec:prop-matching}

We now address the case where the number of items is between $n$ and $2n$.
As discussed at the beginning of Section~\ref{sec:prop-rr}, the round-robin algorithm fails in this regime. Nevertheless, we devise an alternative algorithm that computes a proportional allocation with high probability using some ideas from previous work together with an additional new idea.

\begin{theorem} \label{thm:matching-prop}
Suppose that $\cD$ is PDF-bounded and has mean at most $1/2$.
For $n \leq m \leq 2n$, there is a polynomial-time algorithm that outputs a proportional allocation with high probability.
\end{theorem}

We first explain the intuition leading up to the algorithm.
In prior work of \citet{Suksompong16}, a matching-based algorithm is used to find proportional allocations. The most basic case of the algorithm is the case $m = n$, where the algorithm can be stated as follows:

\begin{algorithm}
\caption{Proportional Algorithm for $m = n$}\label{alg:matching-basic}
\begin{algorithmic}[1]
\Procedure{ThresholdMatching$_\tau(N, M, \{u_i\}_{i\in N})$}{}
\State Let $E_{\geq \tau} \leftarrow \{(i, j) \in N \times M \mid u_i(j) \geq \tau\}$.
\If{$G_{\geq \tau} = (N, M, E_{\geq \tau})$ contains a perfect matching}
\State \Return any perfect matching of $G_{\geq \tau}$
\Else
\State \Return NULL
\EndIf
\EndProcedure
\end{algorithmic}
\end{algorithm}

A classic result of Erd{\H{o}}s and R{\'{e}}nyi states that a random balanced bipartite graph is likely to contain a perfect matching when the probability of each edge occurring is, say, $\frac{1.1\log n}{n}$. We use the notation $\cG(a, b, p)$ to denote a distribution over bipartite graphs where the two vertex sets have size $a$ and $b$, and each edge occurs with probability $p$ independently of other edges.

\begin{lemma}[\citet{ErdosRe64}] \label{lem:er-matching}
Let $G = (A, B, E)$ be a graph sampled from the Erd{\H{o}}s-R{\'{e}}nyi random bipartite graph distribution $\cG(n, n, p)$ where $p = (\log n + \omega(1))/n$. Then, with high probability, $G$ contains a perfect matching.
\end{lemma}

Hence, by setting $\tau = 1 - \frac{1.1 \log n}{\alpha n}$, Algorithm~\ref{alg:matching-basic} produces, with high probability, an allocation in which every agent has utility at least $\tau$. From Chernoff bound, we also know that each agent likely values the whole set $M$ at most $n\left(\frac{1}{2} + \tilde{O}\left(\frac{1}{\sqrt{n}}\right)\right)$. As a result, the algorithm produces a proportional allocation when $m = n$ is sufficiently large.

Notice here that the guaranteed lower bound of $\tau$ on each agent's utility is quite strong. In particular, if, say, $m = 1.999n$ and we just run the above algorithm on the first $m$ items, then the resulting (partial) allocation is already proportional with high probability, because each agent values the whole set $M$ roughly $0.9995n \pm o(n)$.

As a result, we are only left with the case where $m = (2 - o(1))n$. For simplicity of discussion, let us focus on the case $m = 2n - 1$. In this case, Algorithm~\ref{alg:matching-basic} is not yet sufficient for us: recall from the beginning of Section~\ref{sec:prop} that in this regime, there are likely to be agents who need at least two items in order to be proportional. This motivates our algorithm, which consists of two stages. Initially, we run Algorithm~\ref{alg:matching-basic} on the first $n$ items. We then use the remaining $m - n$ items to help ``fix'' the agents for whom the allocation is not yet proportional. In particular, we create a graph where the left vertices correspond to these agents, the right vertices to the remaining $m - n$ items, and there is an edge between an agent and an item exactly when adding that item to the agent's bundle results in the bundle being proportional for that agent. The pseudocode of the algorithm is presented as Algorithm~\ref{alg:matching-two-stage}.

\begin{algorithm}
\caption{Proportional Algorithm for $n \leq m \leq 2n$}\label{alg:matching-two-stage}
\begin{algorithmic}[1]
\Procedure{TwoStageMatching$_\tau(N, M, \{u_i\}_{i\in N})$}{}
\State Let $M^0$ be the first $n$ items in $M$, and let $M^1 = M \setminus M^0$.
\State $(M^0_1, \dots, M^0_n) \leftarrow \textsc{ThresholdMatching}_{\tau}(N, M, \{u_i\}_{i \in N})$.
\If{$(M^0_1, \dots, M^0_n)$ = NULL}
\State \Return NULL
\Else
\State $N^{\text{violated}} \leftarrow \{i \in N \mid u_i(M^0_i) < \frac{u_i(M)}{n}\}$.
\State $E^{\text{fix}} \leftarrow \{(i, j) \in N^{\text{violated}} \times M^1 \mid u_i(j) \geq \frac{u_i(M)}{n} - u_i(M^0_i)\}$.
\If{$G^{\text{fix}} = (N^{\text{violated}}, M^1, E^{\text{fix}})$ contains a matching of size $|N^{\text{violated}}|$}
\State $(M^1_1, \dots, M^1_n) \leftarrow$ the allocation corresponding to the matching
\State \Return $(M^0_1 \cup M^1_1, \dots, M^0_n \cup M^1_n)$.
\Else
\State \Return NULL
\EndIf
\EndIf
\EndProcedure
\end{algorithmic}
\end{algorithm}

Note that when the algorithm does not output NULL, it always outputs a proportional allocation. Hence, to complete the proof of Theorem~\ref{thm:matching-prop}, it suffices to show that it rarely outputs NULL when we set an appropriate value for the parameter $\tau$, which we do in the following lemma.

\begin{lemma} \label{lem:matching-correctness}
Suppose that $\cD$ is an $(\alpha, \beta)$-PDF-bounded distribution with mean at most $1/2$.
For $n \leq m \leq 2n$, Algorithm~\ref{alg:matching-two-stage} with $\tau = 1 - \frac{1.1\log n}{\alpha n}$ outputs NULL with $o(1)$ probability.
\end{lemma}

Before we present a formal proof of Lemma~\ref{lem:matching-correctness}, let us briefly discuss the intuition behind it. First of all, we note that the graph $G^{\text{fix}}$ is very dense---this is because the amount needed to ``fix'' each agent $i$'s bundle (i.e., $u_i(M)/n - u_i(M_i^0) \leq u_i(M)/n - \tau$) is $o(1)$. Hence, intuitively, the matching in $G^{\text{fix}}$ should exist as long as $|N^{\text{violated}}| < |M^1|$. Now, notice that agent $i$ can only belong to $N^{\text{violated}}$ if $u_i(M)/n \geq \tau$. Since $\E[u_i(M)/n] \leq 1$ and $\tau = 1 - O\left(\frac{\log n}{n}\right)$, the probability that this happens for each $i$ should be $1/2 + o(1)$. (A slightly looser, but sufficient, bound will be derived from Lemma~\ref{lem:anti-concen-sum}.) As a result, $|N^{\text{violated}}|$ should be of size only around $n/2$, which is considerably less than $|M^1| = m - n = n - o(n)$.

To formalize the above ideas, we need an additional bound regarding the existence of matchings, which is more tailored towards our application. It says that, if we sample a (non-balanced) random bipartite graph $(A, B, E)$, where $A$ is slightly larger than $B$, with sufficiently large probability $p$, then any subset $S \subseteq A$ of size noticeably smaller than $B$ is likely to contain a matching to $B$. The bound is stated below; note that we do not attempt to optimize any parameters here, and instead only prove a version which suffices for our application. %The proof of the lemma is given in Appendix~\ref{app:matching-proof}.
For a graph $G$ and a set $S$ of vertices, we denote by $Z_G(S)$ the set of vertices adjacent to at least one vertex in $S$.

\begin{lemma} \label{lem:gen-matching}
Let $G = (A, B, E)$ be a graph sampled from the Erd{\H{o}}s-R{\'{e}}nyi random bipartite graph distribution $\cG(n, q, p)$, where $p \geq 0.5$ and $0.9n \leq q \leq n$. Then, with high probability, for every $S \subseteq A$ of size at most $0.6n$, we have $|Z_G(S)| \geq |S|$. 
\end{lemma}

\begin{proof}
We can bound the probability that the ``bad event'' occurs as follows:
\begin{align*}
\Pr[\exists S \subseteq A, |S| \leq 0.6n, |Z_G(S)| < |S|] &\leq \sum_{i=1}^{\lfloor 0.6n\rfloor} \Pr[\exists S \subseteq A, |S| = i, |Z_G(S)| \leq i - 1] \\
&= \sum_{i=1}^{\lfloor 0.6n\rfloor} \Pr[\exists S \subseteq A, T \subseteq B,  |S| = i, |T| = i - 1, Z_G(S) \subseteq T] \\
&\leq \sum_{i=1}^{\lfloor 0.6n\rfloor} \sum_{S \subseteq A, T \subseteq B \atop |S| = i, |T| = i - 1} \Pr[Z_G(S) \subseteq T] \\
&= \sum_{i=1}^{\lfloor 0.6n\rfloor} \sum_{S \subseteq A, T \subseteq B \atop |S| = i, |T| = i - 1} \Pr[\forall u \in S, v \in (B \setminus T), (u, v) \notin E] \\
&= \sum_{i=1}^{\lfloor 0.6n\rfloor} \sum_{S \subseteq A, T \subseteq B \atop |S| = i, |T| = i - 1} \prod_{u \in S, v \in (B \setminus T)} \Pr[(u, v) \notin E] \\
&\leq \sum_{i=1}^{\lfloor 0.6n\rfloor} \sum_{S \subseteq A, T \subseteq B \atop |S| = i, |T| = i - 1} 2^{- i(q + 1 - i)} \\
&= \sum_{i=1}^{\lfloor 0.6n\rfloor} \binom{n}{i}\binom{q}{i - 1} 2^{-i(q + 1 - i)} \\
(\text{From } \binom{n}{i}\leq n^i \text{ and }q \geq 0.9n) &\leq \sum_{i=1}^{\lfloor 0.6n\rfloor} (nq)^i 2^{-0.3i n}. 
\end{align*}
For sufficiently large $n$, we have $2^{0.1n} \geq n^2$. We can use this to further bound the above summation as
\begin{align*}
\sum_{i=1}^{\lfloor 0.6n\rfloor} (nq)^i 2^{-0.3i n} \leq
\sum_{i=1}^{\lfloor 0.6n\rfloor} n^{2i} 2^{-0.3i n} \leq
\sum_{i=1}^{\lfloor 0.6n\rfloor} 2^{-0.2i n} 
\leq 2^{-0.2n} + 0.6n\cdot 2^{-0.4n}
= O(2^{-0.2n}), 
\end{align*}
which concludes our proof.
\end{proof}

With this setup ready, we can now proceed to the proof of Lemma~\ref{lem:matching-correctness}.

\begin{proof}[Proof of Lemma~\ref{lem:matching-correctness}]
Let $G_{\geq \tau} = (N, M^0, E_{\geq \tau})$ be the graph as defined in Algorithm~\ref{alg:matching-basic} with our threshold $\tau = 1 -  \frac{1.1\log n}{\alpha n}$. 
As in the proof of Theorem~\ref{thm:rr-prop}, one can check that Chernoff and union bounds imply that, with high probability, we have
\begin{align} \label{eq:sum-bound-matching}
u_i(M) \leq \frac{m}{2} + 10\sqrt{m \log n}
\end{align}
for all agents $i \in N$.
We will henceforth assume that~\eqref{eq:sum-bound-matching} holds for all $i \in N$. Furthermore, Lemma~\ref{lem:er-matching} immediately implies that a perfect matching in $G_{\geq \tau}$ exists with high probability; we will also assume that this is the case from now on.

 Let $q = m - n$ (so $0 \leq q \leq n$). We consider two cases, based on the value of $q$.

\paragraph{Case 1: $q < 0.9n$.} 
For sufficiently large $n$,~\eqref{eq:sum-bound-matching} implies that $u_i(M)/n \leq 0.996 < \tau$ for all $i \in N$. Hence, the partial allocation $(M^0_1, \dots, M^0_n)$ is already proportional for every agent, and the algorithm outputs this allocation without going into the second stage.

\paragraph{Case 2: $q \geq 0.9n$.}
In this case, we will need to consider two more ``good'' events.
Let $G^1_{\geq 0.5/\beta} = (N, M^1, E^1_{\geq 0.5/\beta})$ be such that $(i, j) \in E^1_{\geq 0.5/\beta}$ if and only if $u_i(j) \geq 0.5/\beta$.
\begin{enumerate}
\item[E1]: $|N^{\text{violated}}| \leq 0.6n$;
\item[E2]:  For every $S\subseteq N$ of size at most $0.6n$, we have $|Z_{G^1_{\geq 0.5/\beta}}(S)|\geq |S|$.
\end{enumerate}

Before we prove that E1 and E2 hold with high probability, let us argue that if they hold, then the algorithm outputs a proportional allocation. 
To this end, first observe that since E1 and E2 hold, Hall's marriage theorem implies that there exists a matching from $N^{\text{violated}}$ to $M^1$ that uses all vertices in $N^{\text{violated}}$. Moreover, notice that for sufficiently large $n$, $u_i(M)/n$, which is at most $1 + o(1)$ by \eqref{eq:sum-bound-matching}, is less than $0.5/\beta + \tau = (1 + 0.5/\beta) - o(1)$ for all $i \in N$. As a result, the aforementioned matching remains a matching in the graph $G^\text{fix}$. The algorithm thus outputs a proportional allocation as desired.

It remains to show that both E1 and E2 occur with high probability. This is obvious for E2, since the graph $G^1_{\geq 0.5/\beta}$ is generated in exactly the same way as in Lemma~\ref{lem:gen-matching} (with $p = 1 - F_{\cD}(0.5/\beta) \geq 0.5$). 

As for E1, let us consider each agent $i \in N$. Let $\sigma^2 > 0$ be the variance of $\cD$. For sufficiently large $n$, Lemma~\ref{lem:anti-concen-sum} implies that 
\begin{align*}
\Pr[u_i(M)/n > \tau]
&= 1 - \Pr[u_i(M) \leq n\tau] \\
&= 1 - \Pr\left[u_i(1) + \cdots + u_i(m) \leq n - (1.1\log n) / \alpha\right] \\
&\leq 1 - \Pr\left[u_i(1) + \cdots + u_i(m) \leq m/2 - 0.01\sigma\sqrt{m}\right] \\
&\leq 0.51 + O(1/\sqrt{m}) \\
&\leq 0.55,
\end{align*}
where the first inequality follows from $n \geq m/2$ and the fact that $1.1 \log n / \alpha = \Theta(\log n)$ is smaller than $0.01 \sigma \sqrt{m} = \Theta(\sqrt{n})$ for any sufficiently large $n$.

By Chernoff bound, with high probability, at most $0.6n$ agents $i$ have $u_i(M)/n > \tau$.
Only these agents can be included into $N^{\text{violated}}$. Hence, E1 holds with high probability.

Thus, the algorithm outputs a proportional allocation with high probability in both cases.
\end{proof}

\section{Envy-freeness up to Any Item}
\label{sec:EFX}

In this section, we turn our attention to an important relaxation of envy-freeness: envy-freeness up to any item (EFX).
While the worst-case existence of EFX is an intriguing open problem, we show that an EFX allocation is likely to exist for \emph{any} relation between the number of agents and the number of items:

\begin{theorem} \label{thm:efx}
Suppose that $\cD$ is PDF-bounded.
There is a polynomial-time algorithm that outputs an EFX allocation with high probability.
\end{theorem}

Let $r = \lfloor m/n \rfloor$ and $q = m - nr$. 
An EFX allocation obviously exists when $m\leq n$ (by assigning at most one item to each agent),\footnote{In fact, \citet{AmanatidisNtMa20} showed that an EFX allocation always exists even when $m\leq n+2$.} so we may restrict our attention to the case $m > n$. 
Furthermore, since an envy-free (and therefore EFX) allocation exists with high probability when $r= \Omega\left(\frac{\log n}{\log \log n}\right)$ (see Section~\ref{sec:EF}), or when $r\geq 2$ and $q=0$ \citep{ManurangsiSu19}, we may assume that $r= O(\log n)$ and $q\geq 1$. 

As with proportionality (Section~\ref{sec:prop}), the existence of EFX allocations will be shown via two kinds of algorithms: round-robin-based and matching-based. The former will work whenever the remainder $q$ is $\omega(1)$. On the other hand, for the case $q = O(1)$, we in fact present two matching-based algorithms: the first works for $r \geq 2$ and the second specifically for $r = 1$. These three algorithms and their corresponding proofs of correctness are given in Sections~\ref{sec:rr-rev}--\ref{sec:EFX-1-small}; we then combine them to deduce Theorem~\ref{thm:efx} in Section~\ref{sec:puttogether-efx}.

\subsection{Round-Robin-Based Algorithm for $q = \omega(1)$}
\label{sec:rr-rev}

We start by describing the round-robin-based algorithm. The algorithm works in exactly the same way as round-robin for the first $r$ rounds: in each round, we let agents $1, 2, \dots, n$ choose their most preferred item in this order. However, in the final round, we reverse the order and let agents $n, n - 1, \dots, n - q + 1$ chooses their most preferred item in this order. %This is in contrast to the ``standard'' round robin, where agent $1, \dots, q$ gets to choose the items in the last round.
\iffalse
The pseudocode of our round-robin algorithm ``with reversed last round'' can be found in Algorithm~\ref{alg:rev-round-robin}.

\begin{algorithm}
\caption{Round-Robin with Reversed Last Round}
\label{alg:rev-round-robin}
\begin{algorithmic}[1]
\Procedure{RoundRobinReversedLast$(N = \{1, \dots, n\}, M, \{\succ_i\}_{i\in N})$}{}
\State $M_i \leftarrow \emptyset$ for all $i \in N$.
\State $M' \leftarrow M$
\For{$t = 1, \dots, r$}
\For{$i = 1, \dots, n$}
\State $j \leftarrow$ Most preferred item of $i$ within $M'$
\State $M_i \leftarrow M_i \cup \{j\}$
\State $M' \leftarrow M' \setminus \{j\}$
\EndFor
\EndFor
\For{$i = n, \dots, n - q + 1$} 
\State $j \leftarrow$ Most preferred item of $i$ within $M'$
\State $M_i \leftarrow M_i \cup \{j\}$
\State $M' \leftarrow M' \setminus \{j\}$
\EndFor
\State \Return $(M_1, \dots, M_n)$
\EndProcedure
\end{algorithmic}
\end{algorithm}
\fi

We show that if $q = \omega(1)$, then this algorithm, which we will refer to as the \emph{round-robin algorithm with reversed last round}, is likely to produce an EFX allocation.

\begin{theorem} \label{thm:rr-efx}
With probability $1 - O(1/\sqrt{q})$, the allocation output by  the round-robin algorithm with reversed last round is EFX.
\end{theorem}

Before we proceed to prove Theorem~\ref{thm:rr-efx}, let us note that using different agent orderings in the algorithm is intuitively a fairer way of distributing items.
For instance, by letting agent $n$ pick first in the last round, we somewhat ``balance out'' the unfairness of the previous rounds. 
On a more formal level, it is also the case that the standard round-robin algorithm fails to give a guarantee as in Theorem~\ref{thm:rr-efx}. 
A simple example is when $r = 1$ and $q = \Omega(n)$. In this case, the output allocation is not EFX for agent $n$ if his most preferred item was picked by one of the first $q$ agents; this bad event happens with constant probability (i.e., $\Omega(q/n)$). However, as Theorem~\ref{thm:rr-efx} shows, reversing the order allows us to rule out such bad events with high probability.

Another remark we would like to make is that when $q \geq 2$ is constant, there is a constant probability that round-robin with reversed last round fails to find an EFX allocation. To see this, consider the case where $r = 1$ (i.e., $m = n + q$). Observe that there is an $\exp(-O(q))$ probability that each of the last $q + 1$ items remaining has value at most, say, $0.1$ for agent $n$. In this case, agent $n$'s bundle has value at most $0.2$ to him/her. On the other hand, there is a constant probability that agent $(n - 1)$'s first item is of value more than $0.2$ to agent $n$. This means that with probability at least $\exp(-O(q))$, agent $n$ would not be EFX; this probability is constant when $q$ is constant.

%The proof of Theorem~\ref{thm:rr-efx} can be found in the appendix.
%At a high level, we show that an agent who receives $r$ items never envies one receiving $r+1$ items; 
%moreover, with high probability, all agents with $r$ items obtain utility at least $r-1$, and all agents with $r+1$ items obtain utility at least $r$.
%Combining these two claims yields the desired result.

We now proceed to the proof of Theorem~\ref{thm:rr-efx}.
To start with, note that a particularly worrying case when proving that EFX is satisfied for an agent is when this agent receives strictly fewer items than some other agent. Our algorithm is specifically designed to handle this issue: the output allocation is always EFX for $i$ with respect to $i'$ for every pair of agents $i,i'$ such that $i$ receives $r$ items and $i'$ receives $r + 1$ items, as stated below. (Note that this guarantee is \emph{not} probabilistic.)

\begin{claim} \label{claim:efx-fewer-items}
For every $i \in \{1, \dots, n - q\}$ and every $i' \in \{n - q + 1, \dots, n\}$, in the allocation output by the round-robin algorithm with reversed last round, $i$ is EFX with respect to $i'$.
\end{claim}

\begin{proof}
As in Section~\ref{sec:EF}, for any agents $i_1, i_2$, we use $X^{i_1, i_2}_t$ to denote agent $i_1$'s utility for the $t$-th item received by agent $i_2$. When $i_1 = i_2$, we abbreviate $X^{i_1, i_2}_t$ as $X^{i_1}_t$. 
First, observe that since agent $i$ chooses his/her most preferred item in each of the first $r$ rounds before $i'$, we have $X^{i}_{\ell} \geq X^{i, i'}_{\ell}, X^{i, i'}_{\ell + 1}$ for all $\ell \in \{1, \dots, r\}$. 

Now, consider any $M'=M_{i'}\backslash\{j\}$ for some item $j\in M_{i'}$. Suppose that $j$ is the item that $i'$ chooses in the $t$-th round. Then, we have
\begin{align*}
u_i(M') &= X^{i, i'}_1 + \cdots + X^{i, i'}_{t - 1} + X^{i, i'}_{t + 1} + \cdots + X^{i, i'}_{r + 1} \\
&\leq X^i_1 + \cdots + X^i_{t - 1} + X^i_{t} + \cdots + X^i_r \\
&= u_i(M_i).
\end{align*} 
As a result, $i$ is EFX with respect to $i'$.
\end{proof}

Next, we demonstrate that agents with the same number of items are EFX with respect to each other. We show this by proving that with high probability, $u_i(M_i) \geq |M_i| - 1$ for all $i$. Notice that if $|M_i| = |M_{i'}|$ for some $i,i'$, then this immediately implies that $u_i(M_i) \geq |M_i| - 1 \geq \max_{j \in M_{i'}} u_i(M_{i'} \setminus \{j\})$, i.e., that $i$ is EFX with respect to $i'$.

The proof that $u_i(M_i) \geq |M_i| - 1$ with high probability is divided into two claims, based on whether $i$ receives $r$ items (Claim~\ref{claim:lb-non-tail}) or $r + 1$ items (Claim~\ref{claim:lb-tail}). The proofs of these claims share some similarities with the proof in Section~\ref{sec:prop} which shows that round-robin produces a proportional allocation with high probability. Nonetheless, the proofs here are slightly different due to the different lower bounds needed and the reversed order in the last round, so we state them in full below.

\begin{claim} \label{claim:lb-non-tail}
With probability $1-O(1/n)$, $u_i(M_i) \geq r - 1$ for all $i \in \{1, \dots, n - q\}$ simultaneously.
\end{claim}

\begin{proof}
We will bound the probability that $u_i(M_i) \geq r - 1$ for a fixed $i \in \{1, \dots, n - q\}$ and apply the union bound in the end.
Observe that, similarly to the standard round-robin algorithm (i.e., Lemma~\ref{lem:rr-simplified-process}), $X_1^i, \dots, X_r^i$ are distributed as follows. First, $X_1^i$ is drawn from $\cD^{\max(m + 1 - i)}$. Then, for each $t = 1, \dots, r - 1$, $X_{t + 1}^i$ is sampled according to $\cD_{\leq X_t^i}^{\max(m + 1 - i - nt)}$.

To prove the desired bound, we use Lemma~\ref{lem:rep-sel} on $X_1^i, \dots, X_{r - 1}^i$, which implies that the following holds with probability at least $1 - 1/n^2$:
\begin{align*}
X_{r - 1}^i \geq 1 - \frac{\beta}{\alpha} \cdot \frac{r\ln(r/p)}{n} = 1 - O\left(\frac{\log^2 n}{n}\right),
\end{align*}
where we use the assumption that $r=O(\log n)$.
When this holds, the bundle from the first $r - 1$ rounds already yields value at least $(r - 1)\left(1 - O\left(\frac{\log^2 n}{n}\right)\right)$ to agent $i$. Hence, in order to have $u_i(M_i) \geq r - 1$, it suffices to have $X_r^i \geq r \cdot O\left(\frac{\log^2 n}{n}\right) = O\left(\frac{\log^3 n}{n}\right)$. Recall that $X_r^i$ is distributed as the maximum of $m+1 - i - n(r - 1) \geq 2q + 1 \geq 3$ random variables independently drawn from $\cD_{\leq X_{r - 1}^i}$. Thus, for $n$ large enough such that $X_{r-1}^i$ is at least, say, $0.5$, Proposition~\ref{prop:conditional-bounded} implies that the probability that $X_r^i < O\left(\frac{\log^3 n}{n}\right)$ is at most $O\left(\frac{\log^3 n}{n}\right)^3 \leq O\left(\frac{1}{n^2}\right)$. As a result, we have $\Pr[u_i(M_i) \geq r - 1] \geq 1 - O(1/n^2)$.

Taking a union bound over all $i \in \{1, \dots, n - q\}$ completes our proof.
\end{proof}

\begin{claim} \label{claim:lb-tail}
With probability $1 - O(1/\sqrt{q})$, $u_i(M_i) \geq r$ for all $i \in \{n - q + 1, \dots, n\}$ simultaneously.
\end{claim}

\begin{proof}
Fix $i \in \{n - q + 1, \dots, n\}$. Similarly to the proof of Claim~\ref{claim:lb-non-tail}, we would like to bound $\Pr[u_i(M_i) < r]$.
To do so, we first use Lemma~\ref{lem:rep-sel} on $X_1^i, \dots, X_{r - 1}^i$, which implies that the following holds with probability at least $1 - 1/n^2$:
\begin{align} \label{eq:lb-third-to-last}
X_{r - 1}^i \geq 1 - \frac{\beta}{\alpha} \cdot \frac{r\ln(r/p)}{n} = 1 - O\left(\frac{\log^2 n}{n}\right).
\end{align}
Moreover, since $X_r^i$ is the maximum of $m + 1 - i - n(r - 1) > q$ random variables independently sampled from $\cD_{\leq X_{r - 1}^i}$, we have
\begin{align*}
\Pr\left[X_r^i < \left(1 - \frac{\beta}{\alpha} \cdot \frac{1}{\sqrt{q}}\right)X^i_{r - 1}\right] \leq \left(1 - \frac{1}{\sqrt{q}}\right)^q \leq e^{-\sqrt{q}},
\end{align*}
where the first inequality follows from Proposition~\ref{prop:conditional-bounded} and the second inequality from the well-known fact that $1+x\leq e^x$ for any real number $x$.

In other words, with probability at least $1 - e^{-\sqrt{q}}$, the following holds:
\begin{align} \label{eq:lb-second-to-last}
X_r^i \geq \left(1 - O\left(\frac{1}{\sqrt{q}}\right)\right)X_{r - 1}^i.
\end{align}
When both~\eqref{eq:lb-third-to-last} and~\eqref{eq:lb-second-to-last} hold, we have
\begin{align*}
X^i_1 + \cdots + X^i_r \geq r - O\left(\frac{r \log^2 n}{n}\right) - O\left(\frac{1}{\sqrt{q}}\right) \geq r - O\left(\frac{1}{\sqrt{q}}\right),
\end{align*}
where the second inequality follows from $r = O(\log n)$ and $q < n$, which means that $O\left(\frac{r \log^2 n}{n}\right) \leq O\left(\frac{\log^3 n}{n}\right) \leq O\left(\frac{1}{\sqrt{n}}\right) \leq O\left(\frac{1}{\sqrt{q}}\right)$.

In other words, to have $u_i(M_i) \geq r$, it suffices to have $X^i_{r + 1} \geq O\left(\frac{1}{\sqrt{q}}\right)$.
Since $X^i_{r + 1}$ is the maximum of $i - (n - q)$ random variables independently sampled from $\cD_{\leq X^i_r}$ (recall that the ordering in the last round is reversed), the probability that it is less than $O\left(\frac{1}{\sqrt{q}}\right)$ is
\begin{align*}
F_{\cD_{\leq X^i_{r}}^{\max(i - (n - q))}}\left(O\left(\frac{1}{\sqrt{q}}\right)\right) 
&= \left(F_{\cD_{\leq X^i_r}}\left(O\left(\frac{1}{\sqrt{q}}\right)\right)\right)^{i - (n - q)} \\
(\text{From Proposition~\ref{prop:conditional-bounded}}) &\leq \left(\frac{\beta}{\alpha} \cdot \frac{O\left(\frac{1}{\sqrt{q}}\right)}{X^i_r} \right)^{i - (n - q)} \\ 
&\leq O\left(\frac{1}{\sqrt{q}}\right)^{i - (n - q)} \\
&\leq O\left(\frac{1}{\sqrt{q}}\right)^{\min\{i - (n - q), 4\}} \\ 
&= O\left(q^{-\min\{0.5(i - (n - q)),2\}}\right),
\end{align*}
where the second inequality holds because, when conditioned on~\eqref{eq:lb-third-to-last} and~\eqref{eq:lb-second-to-last}, we have $X_r^i \geq 0.5$ for any sufficiently large $n, q$.

Hence, in total, we have
\begin{align*}
\Pr[u_i(M_i) < r] \leq \frac{1}{n^2} + \frac{1}{e^{\sqrt{q}}} + O\left(q^{-\min\{0.5(i - (n - q)),2\}}\right) \leq O\left(q^{-\min\{0.5(i - (n - q)),2\}}\right).
\end{align*}
By taking a union bound over $i \in \{n - q + 1, \dots, n\}$, the probability that $u_i(M_i) < r$ for some such $i$ is at most
\begin{align*}
\sum_{i=n-q+1}^n O\left(q^{-\min\{0.5(i - (n - q)),2\}}\right) 
= \sum_{\ell=1}^q O\left(q^{-\min\{0.5 \ell, 2\}}\right) = O(q^{-0.5}),
\end{align*}
which concludes our proof.
\end{proof}

Theorem~\ref{thm:rr-efx} can now be easily proved using the above claims.

\begin{proof}[Proof of Theorem~\ref{thm:rr-efx}]
By Claims~\ref{claim:lb-non-tail} and~\ref{claim:lb-tail}, with probability $1 - O(1/\sqrt{q})$, we have $u_i(M_i) \geq r - 1$ for all $i \in \{1, \dots, n - q\}$ and $u_i(M_i) \geq r$ for all $i \in \{n - q + 1, \dots, n\}$. To see that this implies that the allocation is EFX, let us consider any pair of agents $i, i' \in N$. We argue that $i$ is EFX with respect to $i'$ by considering the following three cases.
\begin{enumerate}
\item $i \geq n - q + 1$. Since $M_{i'}$ contains at most $r + 1$ items, if we remove any item from this bundle, then it has at most $r$ items, meaning that $i$ values it at most $r$. Recall that we have $u_i(M_i) \geq r$. Hence, in this case, $i$ is EFX with respect to $i'$.
\item $i \leq n - q$ and $i' \geq n - q + 1$. Claim~\ref{claim:efx-fewer-items} immediately implies that $i$ is EFX with respect to $i'$.
\item $i, i' \leq n - q$. Similarly to the first case, since $M_{i'}$ contains $r$ items, if we remove any item from this bundle, then it has at most $r - 1$ items, meaning that $i$ values it at most $r - 1$. On the other hand, we have $u_i(M_i) \geq r - 1$. Thus, $i$ is also EFX with respect to $i'$ in this case. \qedhere
\end{enumerate}
\end{proof}

\subsection{An EF-based Algorithm for $r \geq 2$ and $q = O(1)$}
\label{sec:EFX-2-small}

%We move on to our second algorithm, which is a matching-based algorithm. This algorithm can be viewed as a generalization of the algorithm in~\cite{ManurangsiSu19}. There, the following generalization of matching, called \emph{$r$-matching}, was used: An $r$-matching of a bipartite graph $G$ is a subgraph of $G$ such that every left vertex has degree at most $r$ and every right vertex has degree at most 1. Here, we will use a generalization where the left degrees can be different.

%\begin{lemma}
%Let $G = (N, M, E)$ be a bipartite graph and let $\br = (r_i)_{i \in N}$ be an $n$-dimensional vector of non-negative integers. An \emph{$\br$-matching} of $G$ is a subgraph of $G$ such that each left vertex $i \in N$ has degree at most $i$ and each right vertex has degree at most one. An $\br$-matching of $G$ is said to be \emph{perfect} if .
%\end{lemma}

We move on to our second algorithm, which handles the case where $r\geq 2$ while the remainder $q$ is $O(1)$. 
In this case, we will use the algorithm of \citet{ManurangsiSu19} as a subroutine. 
The algorithm there works when $m$ is divisible by $n$ and produces an envy-free allocation with high probability. 
Furthermore, it guarantees that every item is valued at least $1 - O\left(\frac{\log m}{n}\right)$ by the agent it is assigned to.\footnote{The guarantee as stated in~\cite{ManurangsiSu19} is that this value is at least $1 - O\left(\frac{\log m}{n}\right)^{1/q}$ when $\cD$ is $(\underline{\theta}, \overline{\theta}, q)$-polynomially bounded at $1$ (see the definition in their paper). However, our $(\alpha, \beta)$-PDF-boundedness assumption implies that $\cD$ is $(\alpha, \beta, 1)$-polynomially bounded at $1$, which yields our claimed bound.} This is stated more formally below.

\begin{theorem}[\citet{ManurangsiSu19}] \label{thm:ef-blackbox}
When $m$ is divisible by $n$ and $2n \leq m \leq 2^{o(n)}$, there exists an algorithm $\cA$ that, with high probability, outputs an envy-free allocation $(M_1, \dots, M_n)$ such that $|M_1| = \dots = |M_n|$ and $u_i(j) \geq 1 - O\left(\frac{\log m}{n}\right)$ for all $i \in N$ and $j \in M_i$.
\end{theorem}

The algorithm $\cA$ in Theorem~\ref{thm:ef-blackbox} is a matching-based algorithm. We will not give the full description of the algorithm here since we do not need it, and instead simply use $\cA$ in a black-box manner.

Our EFX algorithm is incredibly simple: we just run $\cA$ on the first $rn$ items in $M$. The rest of the items are assigned arbitrarily, in such a way that each agent receives at most one item. The pseudocode of the algorithm is given as Algorithm~\ref{alg:matching-simple}.

\begin{algorithm}
\caption{EFX Algorithm for $r\geq 2$ and $q=O(1)$}
\label{alg:matching-simple}
\begin{algorithmic}[1]
\Procedure{EFX-via-EF$(N = \{1, \dots, n\}, M, \{u_i\}_{i\in N})$}{}
\State $M^0 \leftarrow$ the set of first $rn$ items in $M$
\State $M^1 \leftarrow M \setminus M^0$
\State $(M_1^0, \dots, M_n^0) \leftarrow$ $\cA(N, M^0, \{u_i\}_{i \in N})$ \label{line:ef-blackbox}
\For{$i = n - q + 1, \dots, n$}
\State $j_i \leftarrow$ arbitrary item from $M^1$ 
\State $M^1 \leftarrow M^1 \setminus \{j_i\}$
\EndFor
\State \Return $(M^0_1, \dots, M^0_{n - q}, M^0_{n - q + 1} \cup \{j_{n - q + 1}\}, \dots, M^0_n \cup \{j_n\})$
\EndProcedure
\end{algorithmic}
\end{algorithm}

The main result of this subsection is that Algorithm~\ref{alg:matching-simple} produces an EFX allocation with high probability when $r \geq 2$ and $q=O(1)$ (in fact, even when $q = o\left(\frac{\sqrt{n}}{\log^3 n}\right)$). 
This follows from the theorem that we state next and our assumption that $r=O(\log n)$ (which implies $m=O(n\log n)$).

\begin{theorem} \label{thm:efx-blackbox}
When $r \geq 2$, Algorithm~\ref{alg:matching-simple} outputs an EFX allocation with probability $1 - o(1) - O\left(\frac{q^2 r^3 \log^3 m}{n}\right)$.
\end{theorem}

Before we proceed to the proof, let us describe its high-level idea. Let $\tau := 1 - O\left(\frac{\log m}{n}\right)$ be the lower bound on the utility guaranteed by Theorem~\ref{thm:ef-blackbox}.  
The theorem ensures that each agent $i$ receives a bundle that he/she values at least $r\tau$ from $\cA$, and that the partial allocation $(M_1^0, \dots, M_n^0)$ is envy-free. (Note that here we use the assumption $r \geq 2$, which is required by Theorem~\ref{thm:ef-blackbox}.) 

Since the partial allocation $(M_1^0, \dots, M_n^0)$ is envy-free and every item yields value at most $1$, agent $i$ will be EFX with respect to agent $i'$, unless in the second step $i'$ receives an item $j$ with $u_i(j) \geq r\tau - (r - 1) = 1 - O\left(\frac{r \log m}{n}\right)$---call this latter event (*). This is a low-probability event for a fixed pair $i, i'$; however, it cannot be neglected when we consider all pairs of agents, because every item is likely to be valued more than $1 - O\left(\frac{r \log m}{n}\right)$ by multiple agents. 

To make the proof work, we need to make the following additional observation: since $r \geq 2$, if $i$ is not EFX with respect to $i'$, it must also be the case that there exists $j \in M^0_{i'}$ such that $u_i(j) \geq r\tau - (r - 1)$---call this event (**). %\footnote{Indeed, assume that $u_i(j) < r\tau - (r-1)$ for all $j\in M^0_{i'}$.
%Since $r\geq 2$, even after removing an item from the bundle of agent $i'$, at least one item from $M^0_{i'}$ must remain.
%This item has value less than $r\tau - (r-1)$ for agent $i$, and therefore the remaining bundle has value at most $(r\tau - (r-1)) + (r-1) = r\tau$ for $i$, implying that $i$ is EFX with respect to $i'$.}
One can check that the probability that both (*) and (**) occur together is only $\frac{(q r \log m)^{O(1)}}{n^2}$, meaning that we may now use the union bound over all pairs $i, i'$ (with $i' \geq n - q$) to finish the proof.

The proof below follows the outlined approach; in particular, we refer to a pair $i, i'$ that may violate (**) as a \emph{potentially problematic} pair, and a pair $i, i'$ that may violate both (*) and (**) as a \emph{problematic} pair. The main contribution in the formal proof below is to show that with high probability, no problematic pair exists.

% and (2) there exists some  $$ $1 - $.

\begin{proof}[Proof of Theorem~\ref{thm:efx-blackbox}]
Let $\tau := 1 - O\left(\frac{\log m}{n}\right)$ be the lower bound on the utility guaranteed by Theorem~\ref{thm:ef-blackbox}, and let $\tau' := r\tau - (r - 1) = 1 - O\left(\frac{r \log m}{n}\right)$. Now, for every agent $i \in N$ and $i' \in \{n - q + 1, \dots, n\} \setminus \{i\}$, we say that the pair $(i, i')$ is \emph{potentially problematic} if there exists an item $j \in M^0$ such that $u_i(j) \geq \tau'$ and $u_{i'}(j) \geq \tau$. Furthermore, an agent $i \in N$ is said to be \emph{potentially problematic} if $(i, i')$ is potentially problematic for some $i' \in \{n - q + 1, \dots, n\} \setminus \{i\}$.

Let us fix $i \in N$ and $i' \in \{n - q + 1, \dots, n\} \setminus \{i\}$. Consider any item $j \in M$. The probability that $u_i(j) \geq \tau'$ and $u_{i'}(j) \geq \tau$ is at most $\beta(1 - \tau) \cdot \beta(1 - \tau') = O\left(\frac{r \log^2 m}{n^2}\right)$. We can use a union bound on all items $j \in M^0$ to derive
\begin{align*}
\Pr[(i, i') \text{ is potentially problematic}] \leq m \cdot O\left(\frac{r \log^2 m}{n^2}\right) = O\left(\frac{r^2 \log^2 m}{n}\right).
\end{align*}
Hence, by once again taking a union bound over $i' \in \{n - q + 1, \dots, n\} \setminus \{i\}$, we have
\begin{align} \label{eq:potentially-prob}
\Pr[i \text{ is potentially problematic}] \leq O\left(\frac{q r^2 \log^2 m}{n}\right).
\end{align}

Now, for each agent $i \in N$, we say that $i$ is \emph{problematic} if $i$ is potentially problematic and there exists $j \in M^1$ such that $u_i(j) \geq \tau'$. Let us bound the probability that the latter happens. To do so, note that for a fixed $j \in M^1$, $\Pr[u_i(j) \geq \tau'] \leq \beta(1 - \tau') = O\left(\frac{r \log m}{n}\right)$. Hence, by union bound over all $q$ items in $M^1$, we have
\begin{align} \label{eq:high-val-m1}
\Pr[\exists j \in M^1, u_i(j) \geq \tau'] \leq O\left(\frac{q r \log m}{n}\right).
\end{align}
Notice that the two events ``$i$ is potentially problematic'' and ``there exists $j \in M^1$ such that $u_i(j) \geq \tau'$'' are independent, because the former only concerns valuations of items in $M^0$ whereas the latter concerns those in $M^1$. As a result, by combining~\eqref{eq:potentially-prob} and~\eqref{eq:high-val-m1}, we have
\begin{align*}
\Pr[i \text{ is problematic}] \leq O\left(\frac{q^2 r^3 \log^3 m}{n^2}\right).
\end{align*}
Applying the union bound over all $i \in N$ yields
\begin{align*}
\Pr[\exists i \in N, i \text{ is problematic}] \leq O\left(\frac{q^2 r^3 \log^3 m}{n}\right).
\end{align*}

Finally, from Theorem~\ref{thm:ef-blackbox} and the assumption $r \geq 2$, we have that with high probability, the partial allocation $(M^0_1, \dots, M^0_n)$ is envy-free, $|M^0_1| = \dots = |M^0_n| = r$, and each agent values each assigned item at least $\tau$. Assume that this is the case, and that there is no problematic $i\in N$.
To conclude the proof, it suffices to show that the (complete) allocation produced by Algorithm~\ref{alg:matching-simple} is EFX.
Consider any pair of distinct agents $i, i'$, and divide into two cases:
\begin{enumerate}
\item $i' \leq n - q$. Since $i'$ does not receive an item in the second phase of the algorithm, $i$ does not envy $i'$ due to the guarantee from Theorem~\ref{thm:ef-blackbox}.
\item $i' > n - q$. Since $i$ is not problematic, we know that either $i$ is not potentially problematic or $u_i(j) < \tau'$ for all $j \in M^1$. We analyze these two cases separately.
\begin{enumerate}
\item $i$ is not potentially problematic. In this case, $(i, i')$ is not potentially problematic. Since $u_{i'}(j) \geq \tau$ for all $j \in M^0_{i'}$, we must have $u_{i}(j) < \tau'$ for all $j \in M^0_{i'}$. Recall that $r \geq 2$, which means that even after removing any item from the bundle $M_{i'}$ of $i'$, at least one item from $M^0_{i'}$ remains.
Thus, after removing any item from $M_{i'}$, the bundle is valued by $i$ less than $(r - 1) + \tau' = r \cdot \tau$. Since $i$ values her own bundle at least $r \cdot \tau$, we have that $i$ is EFX with respect to $i'$.

\item $u_i(j) < \tau'$ for all $j \in M^1$. Consider a bundle $M'$ that results from removing one item from $M_{i'}$. If the removed item is the item that $i'$ receives last, then $M'$ is exactly $M^0_{i'}$; due to the envy-freeness guarantee of $(M^0_1, \dots, M^0_n)$, we have $u_i(M_i) \geq u_i(M')$. On the other hand, if the item that $i'$ receives last is not removed from the bundle, this item is valued less than $\tau'$ by agent $i$. This means that $u_i(M') < (r - 1) + \tau' = r \cdot \tau \leq u_i(M_i)$. Hence, we can again conclude that $i$ is EFX with respect to $i'$. \qedhere
\end{enumerate}
\end{enumerate}
\end{proof}

\subsection{A Matching-Based Algorithm for $r = 1$ and $q = O(1)$}
\label{sec:EFX-1-small}

We now proceed to our final case: $r = 1$ and $q = O(1)$ (i.e., $m = n + O(1)$). The algorithm for this case is inspired by that from the previous subsection. At a high level, we again use a matching-based algorithm to first assign one item to each agent while ensuring that each agent highly values his/her own item. Then, in the second step, we give an additional item to some agents. Although this general outline appears very similar to Algorithm~\ref{alg:matching-simple}, we have to be much more careful here: since the starting assignment is no longer envy-free, we cannot simply select arbitrary agents to receive additional items in the second step. In particular, if agent $i$ envies agent $i'$ after the first step and agent $i'$ receives an additional item, then agent $i$ must also receive an additional item for the allocation to be EFX.

To describe our algorithm more precisely, let us introduce some additional notation. First, recall that an \emph{assignment} $\psi$ is simply an injection from $N$ to $M$. For a weight function $w: N \times M \to \R$, the \emph{weight} of an assignment $\psi$ is $\sum_{i \in N} w(i, \psi(i))$. It is well-known that a maximum weight assignment can be found in polynomial time~\cite{Kuhn55,Munkres57}. Another concept that will be useful is a \emph{topological ordering} of a directed acyclic graph $(V, E)$, which is defined as a mapping $\sigma: V \to \{1,2,\dots,|V|\}$ such that for any edge $u \to v$ in $E$, it holds that $\sigma(u) < \sigma(v)$. Again, it is well-known that such a ordering can be computed efficiently for any directed acyclic graph~\cite{Kahn62}. Finally, following~\citet{LiptonMaMo04}, we define the \emph{envy graph} of an assignment $\psi$ to be the graph for which the vertex set is the set of all agents, and there is an edge from agent $i$ to agent $i'$ if and only if $u_i(\psi(i)) < u_i(\psi(i'))$.

With these prerequisites in mind, we now describe our algorithm. First, for an appropriate threshold $\tau$ (to be chosen later) we define a weight function $w$ such that $w(i, j) = u_i(j)$ if $u_i(j) \geq \tau$ and $w(i, j) = -\infty$ otherwise, and find a maximum-weight assignment $\psi$ with respect to $w$. We then create the envy graph for the assignment $\psi$ and assign one of the $q$ unused items to each of the $q$ lowest-ranked agents in the topological ordering associated with the envy graph. The pseudocode of the algorithm is presented as Algorithm~\ref{alg:max-matching}.

\begin{algorithm}[h!]
\caption{EFX Algorithm for $r=1$ and $q=O(1)$}
\label{alg:max-matching}
\begin{algorithmic}[1]
\Procedure{MaximumAssignment$_{\tau}(N, M, \{u_i\}_{i\in N})$}{}
\State For every $i \in N, j \in M$, let
\begin{align*}
w(i, j) =
\begin{cases}
u_i(j) &\text{ if } u_i(j) \geq \tau; \\
-\infty &\text{ otherwise.}
\end{cases}
\end{align*}
\State $\psi \leftarrow$ maximum weight assignment with respect to $w$.
\If{$w(i, \psi(i)) = -\infty$ for some $i \in N$}
\State \Return NULL \label{step:efx-no-matching}
\EndIf
\State $E_{\text{envy}} \leftarrow \{(i, i') \in N \times N \mid u_i(\psi(i)) < u_i(\psi(i'))\}$. \label{step:envy-graph-creation}
\State $\sigma \leftarrow$ arbitrary topological ordering of $(N, E_{\text{envy}})$. \label{step:topological-ordering}
\State $\{j_1, \dots, j_q\} \leftarrow$ items not used in $\psi$. \label{step:unused-items}
\For{$i \in N$}
\If{$\sigma(i) \leq q$}
\State $M_i \leftarrow \{\psi(i), j_{\sigma(i)}\}$
\Else
\State $M_i \leftarrow \{\psi(i)\}$
\EndIf
\EndFor
\State \Return $(M_1, \dots, M_n)$
\EndProcedure
\end{algorithmic}
\end{algorithm}

Before we prove the correctness of the algorithm, let us argue that the algorithm is even \emph{valid}. In particular, Line~\ref{step:topological-ordering} implicitly assumes that the graph $(N, E_{\text{envy}})$ is acyclic. We show below that this always holds.

\begin{lemma}
\label{lem:efx-no-cycle}
$(N, E_{\text{envy}})$ does not contain a cycle.
\end{lemma}

\begin{proof}
Suppose for the sake of contradiction that $(N, E_{\text{envy}})$ contains a cycle $i_1 \to i_2 \to \cdots \to i_t \to i_1$. Due to Line~\ref{step:efx-no-matching}, if the algorithm reaches Line~\ref{step:topological-ordering}, then it must be that $w(i, \psi(i)) = u_i(\psi(i)) \geq \tau$ for all $i \in N$. Hence, the cycle implies that $w(i_1, \psi(i_1)) + \cdots + w(i_{t - 1}, \psi(i_{t - 1})) + w(i_t, \psi(i_t)) < w(i_1, \psi(i_2)) + \cdots + w(i_{t - 1}, \psi(i_t)) + w(i_t, \psi(i_1))$. In other words, if we adjust the assignment $\psi$ so that $i_1$ receives $\psi(i_2)$, $i_2$ receives $\psi(i_3)$, \dots, and $i_t$ receives $\psi(i_1)$, then this would result in a higher total weight, which is a contradiction to the definition of $\psi$. 
\end{proof}

Now that we have established the validity of the algorithm, we show that the algorithm outputs an EFX allocation with high probability if we choose $\tau=1 - \frac{2\log n}{\alpha n}$.

\begin{theorem} \label{thm:efx-max-matching}
For $q = o(n / \log n)$ and $\tau=1 - \frac{2\log n}{\alpha n}$, Algorithm~\ref{alg:max-matching} outputs an EFX allocation with high probability.
\end{theorem}

To prove Theorem~\ref{thm:efx-max-matching}, it is helpful to clarify the independence between the different steps of our algorithm. In this regard, let us think of each random variable $u_i(j)$ as being generated using three independent random variables:
\begin{itemize}
\item A Bernoulli random variable $b_{i, j}$ such that $b_{i, j} = 1$ with probability $F_{\cD}(\tau) \leq \max\{0, 1 - \frac{2\log n}{n}\}$.
\item A random variable $v_{i, j}^{\text{high}}$ sampled from $\cD_{> \tau}$.
\item A random variable $v_{i, j}^{\text{low}}$ sampled from $\cD_{\leq \tau}$.
\end{itemize}
Once these three random variables are sampled, if $b_{i, j} = 1$, then $u_i(j)$ is set to $v_{i, j}^{\text{low}}$. Otherwise, $u_{i}(j)$ is set to $v_{i, j}^{\text{high}}$.

By viewing the utilities as generated by the above process, it is obvious that our algorithm up until Line~\ref{step:unused-items} depends only on $\{b_{i, j}\}_{i \in N, j \in M}$ and $\{v_{i, j}^{\text{high}}\}_{i \in N, j \in M}$ (but \emph{not} on $\{v_{i, j}^{\text{low}}\}_{i \in N, j \in M}$). With this in mind, we proceed with our analysis in two stages. The first stage is up until Line~\ref{step:unused-items}, for which we use the randomness of $\{b_{i, j}\}_{i \in N, j \in M}$ to upper bound the probability of the algorithm terminating at Line~\ref{step:efx-no-matching}, as formalized below.

\begin{lemma} \label{lem:matching-exists}
With high probability (over the randomness of $\{b_{i, j}\}_{i \in N, j \in M}$), Algorithm~\ref{alg:max-matching} does not return NULL.
\end{lemma}

\begin{proof}
Let $M'$ be the set of the first $n$ items.
Consider the bipartite graph $G = (N, M', E)$ where $E = \{(i, j) \in N \times M' \mid b_{i, j} = 0\}$. 
If the graph $G$ contains a perfect matching, then the algorithm does not return NULL (at line~\ref{step:efx-no-matching}); this is because if we pick the assignment $\psi$ that corresponds to this perfect matching, then $\psi$ has a positive weight. Observe also that $G$ is distributed as $\cG(n, n, p)$ where $p = \Pr[b_{i, j} = 0] = 1 - F_{\cD}(\tau) \geq \min\{1, \frac{2\log n}{n}\}$. Thus, from Lemma~\ref{lem:er-matching}, $G$ contains a perfect matching with high probability.
\end{proof}

Next, we consider the second stage of the algorithm, which is after Line~\ref{step:unused-items}. For this purpose, we may think of $\{b_{i, j}\}_{i \in N, j \in M}$ and $\{v^{\text{high}}_{i, j}\}_{i \in N, j \in M}$ as arbitrary (i.e., worst case) and $\{v^{\text{low}}_{i, j}\}_{i \in N, j \in M}$ as being random. 
We will prove two more lemmas. The first lemma is that the output allocation is EFX for all agents whose topological rank is at least $q + 1$. We note that this lemma is \emph{not} probabilistic and holds regardless of the values of the random variables.

\begin{lemma} \label{lem:efx-non-prob}
When Algorithm~\ref{alg:max-matching} does not output NULL, the output allocation is EFX for all agents $i \in N$ such that $\sigma(i) > q$.
\end{lemma}

\begin{proof}
Fix $i \in N$ such that $\sigma(i) > q$, and consider any agent $i' \ne i$. If $i'$ also satisfies $\sigma(i') > q$, then $i'$ receives only one item and hence $i$ is EFX with respect to $i'$. On the other hand, if $i'$ satisfies $\sigma(i') \leq q$, then $i'$ receives two items $\psi(i')$ and $j_{\sigma(i')}$. Since $\sigma(i) > q \geq \sigma(i')$, there is no edge from $i$ to $i'$ in $E_{\text{envy}}$, meaning that
$u_i(\psi(i)) \geq u_i(\psi(i'))$. Furthermore, since $j_{\sigma(i')}$ is unused in the assignment $\psi$, it must be that $u_i(\psi(i)) \geq u_i(j_{\sigma(i')})$, as otherwise changing $\psi(i)$ to $j_{\sigma(i')}$ would increase the weight of $\psi$. This means that $i$ values $M_i = \{\psi(i)\}$ at least as much as each of the two items received by $i'$; hence, $i$ is again EFX with respect to $i'$. 
\end{proof}

The other lemma that we need is that, with high probability, the output allocation is EFX for all agents whose topological rank is at most $q$. Note that this probability is over $\{v^{\text{low}}_{i, j}\}_{i \in N, j \in M}$, and the lemma holds for any (i.e., worst-case) values of $\{b_{i, j}\}_{i \in N, j \in M}$ and $\{v^{\text{high}}_{i, j}\}_{i \in N, j \in M}$.

\begin{lemma} \label{lem:efx-prob}
When Algorithm~\ref{alg:max-matching} does not output NULL, with probability $1 - O(q \log n / n)$ (over the randomness of $\{v_{i, j}^{\text{low}}\}_{i \in N, j \in M}$), the output allocation is EFX for all $i \in N$ such that $\sigma(i) \leq q$.
\end{lemma}

\begin{proof}
We will argue that for each $i \in N$ such that $\sigma(i) \leq q$, the probability that the output allocation is \emph{not} EFX for $i$ is $O(\log n/n)$. Applying the union bound over all $i \in \{\sigma^{-1}(1), \dots, \sigma^{-1}(q)\}$ yields the desired result.

In fact, we will prove an even stronger claim that for each $i \in N$ with $\sigma(i) \leq q$, we have $u_i(M_i) \geq 1$ with probability $1 - O(\log n/n)$. Note that since each agent receives at most two items, $u_i(M_i) \geq 1$ immediately implies that the allocation is EFX for $i$.

To prove our claim, we consider two cases based on whether $b_{i, j_{\sigma(i)}} = 1$.
\begin{enumerate}
\item $b_{i, j_{\sigma(i)}} = 0$. In this case, we have $u_i(j_{\sigma(i)}) \geq \tau$. Furthermore, from our assumption that the algorithm does not output NULL, we have $u_i(\psi(i)) \geq \tau$. Hence, we have $u_i(M_i) \geq 2\tau = 2 - O(\log n / n)$, which is at least 1 for any sufficiently large $n$.
\item $b_{i, j_{\sigma(i)}} = 1$. Once again, from the assumption that the algorithm does not output NULL, we have $u_i(\psi(i)) \geq \tau$. Hence, if $u_i(M_i) < 1$, we must have $u_i(j_{\sigma(i)}) < 1 - \tau$, which happens with probability
\begin{align*}
\Pr[v_{i, j_{\sigma(i)}}^{\text{low}} < 1 - \tau] = \frac{F_{\cD}(1 - \tau)}{F_{\cD}(\tau)} \leq \frac{\frac{2\beta\log n}{\alpha n}}{1 - \frac{2\beta\log n}{\alpha n}} = O\left(\frac{\log n}{n}\right).
\end{align*}
\end{enumerate}
Hence, in both cases, we have $u_i(M_i) \geq 1$ with probability at least $1 - O\left(\frac{\log n}{n}\right)$, completing our proof.
\end{proof}

Finally, by combining Lemmas~\ref{lem:matching-exists}, \ref{lem:efx-non-prob} and~\ref{lem:efx-prob}, we immediately arrive at Theorem~\ref{thm:efx-max-matching}.

\subsection{Putting Things Together: Proof of Theorem~\ref{thm:efx}}
\label{sec:puttogether-efx}

The main theorem of this section (Theorem~\ref{thm:efx}) can now be established by simply selecting one of the three algorithms based on the range of the parameters.

\begin{proof}[Proof of Theorem~\ref{thm:efx}]
If $q=\omega(1)$, we run the round-robin algorithm with reversed last round; from Theorem~\ref{thm:rr-efx}, it outputs an EFX allocation with high probability.
Else, $q=O(1)$. 
If $r \geq 2$, we run Algorithm~\ref{alg:matching-simple}; from Theorem~\ref{thm:efx-blackbox}, this outputs an EFX allocation with high probability.
Finally, if $r=1$, we run Algorithm~\ref{alg:max-matching}, which, from Theorem~\ref{thm:efx-max-matching}, yields an EFX allocation with high probability.
\end{proof}

\section{Envy-free Assignments}
\label{sec:assignments}

In this section, we address another important resource allocation setting: assignments.
Unlike with allocations, here we assign exactly one item to every agent and leave the remaining items unassigned.
Envy-freeness in the assignment setting means that every agent values her assigned item at least as much as that of any other agent.
Since agents only compare individual items, the (non-atomic) distribution $\cD$ no longer plays an important role, and we may simply assume that each agent draws a strict ranking of items from most preferred to least preferred independently of other agents.
Recently, \citet{GanSuVo19} showed that  an envy-free assignment is likely to exist if $m=\Omega(n\log n)$, thereby leaving a gap between $\Omega(n\log n)$ and $n$ (the latter is the minimum number of items needed so that any feasible assignment exists).
Our contribution is the following theorem, which essentially closes this gap.

\begin{theorem} \label{thm:ef-assignment}
Let $\varepsilon > 0$ be any constant.
If $\frac{m}{n} \geq e + \varepsilon$, then with high probability an envy-free assignment exists. On the other hand, if $\frac{m}{n} \leq e - \varepsilon$, then with high probability no envy-free assignment exists. 
\end{theorem}

Our proof of Theorem~\ref{thm:ef-assignment} follows an approach pioneered by \citet{KarpSi81} and later expanded upon by \citet{Wormald95} and others. Roughly speaking, this method can be applied to analyze greedy algorithms when the input is randomized. To apply the method, we first write out (probabilistic) recurrence relations for certain quantities important to the algorithm. Secondly, we convert these recurrence relations into continuous ones (i.e., differential equations), which we can then solve. The final step is to use concentration inequalities to show that the continuous solution and the discrete one are approximately the same with high probability; from this discrete solution, we can then deduce how well the algorithm performs. For more details and examples on this approach, we refer to the survey of \citet{Wormald99}.

The rest of this section applies the outlined method to our problem. In particular, we start by describing a simple greedy algorithm for the problem in Section~\ref{sec:ef-assignment-alg}. Then, we write out the probabilistic recurrence relations in Section~\ref{sec:recurrence-discrete} and translate them to their continuous counterpart in Section~\ref{sec:diff-eq}.
%Finally, Theorem~\ref{thm:ef-assignment} can be shown by putting all the pieces together.
Finally, we put all the pieces together and prove Theorem~\ref{thm:ef-assignment} in Section~\ref{sec:proof-ef-assignment}.

\subsection{Greedy Algorithm}
\label{sec:ef-assignment-alg}

We first present a simple greedy algorithm which, as long as there are no ties in the ranking, produces a correct answer: it outputs an envy-free assignment if one exists, and NULL otherwise.
The algorithm starts with an empty assignment and marks every item as ``valid''. While there is still at least one unassigned agent and at least one valid item left, it picks an unassigned agent arbitrarily and considers her most preferred item among the valid items. If this item is not yet matched to any agent, then assign it to this agent. Otherwise, mark this item invalid and deassign it from any agent it was assigned to.
The algorithm then terminates with an envy-free assignment if there is at least one valid item at the end (in which case there must also be at least $n$ such items), and with no assignment otherwise.

The pseudocode for the algorithm is presented below; here we use $\perp$ to represent ``unassigned'' in the assignment, and $M'$ and $\succ$ to denote the set of ``valid'' items and a ranking, respectively.

\begin{algorithm}
\caption{Greedy Algorithm for Envy-Free Assignment}
\label{alg:greedy-assignment}
\begin{algorithmic}[1]
\Procedure{GreedyAssignment$(N, M, \{\succ_i\}_{i\in N})$}{}
\State $\psi(i) \leftarrow \perp$ for all $i \in N$.
\State $M' \leftarrow M$
\While{$\psi(i) = \perp$ for some $i \in N$ and $M' \ne \emptyset$}
\State $j \leftarrow$ most preferred item of $i$ within $M'$ \label{line:pick-most-pref}
\If{$\psi(i') = j$ for some $i' \in N$}
\State $\psi(i') \leftarrow \perp$ \label{line:unassign}
\State $M' \leftarrow M' \setminus \{j\}$ \label{line:invalidate}
\Else
\State $\psi(i) \leftarrow j$ \label{line:new-assign}
\EndIf
\EndWhile
\If{$M' \ne \emptyset$}
\State \Return $\psi$
\Else
\State \Return NULL
\EndIf
\EndProcedure
\end{algorithmic}
\end{algorithm}

The following lemma establishes the correctness of the algorithm.

\begin{lemma} \label{lem:greedy-correctness}
When the rankings $\{\succ_i\}_{i \in N}$ contain no ties, Algorithm~\ref{alg:greedy-assignment} always outputs an envy-free assignment if one exists, and NULL otherwise.
\end{lemma}

\begin{proof}
First, observe that if the algorithm outputs an assignment $\psi$, the assignment must be envy-free. Indeed, the items are assigned in such a way that for each agent $i$, $\psi(i)$ is the most preferred item of $i$ within $M'$.

Hence, it remains to show that when Algorithm~\ref{alg:greedy-assignment} outputs NULL, no envy-free assignment exists. Suppose that the algorithm indeed outputs NULL, and let $j_1, \dots, j_m$ denote the items in the order that they are removed from $M'$. We will use induction to show that if we assign one of $j_1, \dots, j_m$ to an agent, then the assignment is not envy-free.
This in turn implies that no envy-free assignment exists.

Before we proceed to our inductive proof, let us introduce an additional notation: for every $t \in [m]$, let $i_t$ and $i'_t$ denote the agents $i$ and $i'$ that result in the removal of $j_t$ from $M'$ in line~\ref{line:invalidate} of the algorithm.

\textbf{Base Case.} Consider any assignment such that $j_1$ is assigned to some agent. Since $j_1$ is the most preferred items for both $i_1$ and $i'_1$, at least one of these two agents will envy the agent who receives $j_1$. Hence, the assignment cannot be envy-free.

\textbf{Inductive Step.} Suppose that, for some $1 \leq t < m$, any assignment that uses at least one of $j_1, \dots, j_t$ is not envy-free. Consider any assignment that uses $j_{t + 1}$. If this assignment uses any of $j_1, \dots, j_t$, it cannot be envy-free by our inductive hypothesis, so we may assume that the assignment does not use any of $j_1, \dots, j_t$. 
Given how the algorithm works, it must be the case that $j_{t + 1}$ is the most preferred item for both $i_{t + 1}$ and $i'_{t + 1}$ among all the items in $M \setminus \{j_1, \dots, j_t\}$. Hence, at least one of these two agents will envy the agent who receives $j_{t + 1}$, which means that the assignment is again not envy-free. 
This completes the inductive step and our proof.
\end{proof}

\subsection{Recurrence Relations}
\label{sec:recurrence-discrete}

Having described a greedy algorithm for the problem, we now write down a recurrence relation corresponding to the algorithm. To do so, let us use $M'_t$ to denote the set $M'$ after the $t$-th iteration of the while loop. We also use $\psi_t$ to denote the (possibly partial) assignment after the $t$-th iteration and $Y_t$ to denote the number of items assigned in $\psi_t$; equivalently, $Y_t = |\{i \in N \mid \psi(i) \ne \perp\}|$. We let $X_t$ denote the number of unassigned items in $M'_t$ (with respect to $\psi_t$), i.e., $X_t = |M'_t| - Y_t$. Initially, we have $X_0 = m$ and $Y_0 = 0$.

At step $t + 1 \geq 1$ with $Y_t < n$, notice that conditioned on the current set of valid items $M'_t$ and the current assignment $\psi_t$, the item $j$ picked in Line~\ref{line:pick-most-pref} is distributed uniformly at random among all items in $M'_t$. Hence, the probability that this item is not used in $\psi_t$ is $\frac{X_t}{X_t + Y_t}$; when the item is not used, the algorithm goes to Line~\ref{line:new-assign} and we have $(X_{t + 1}, Y_{t + 1}) = (X_t - 1, Y_t + 1)$. On the other hand, with probability $\frac{Y_t}{X_t + Y_t}$, the algorithm executes Lines~\ref{line:unassign} and~\ref{line:invalidate}, resulting in $(X_{t + 1}, Y_{t + 1}) = (X_t, Y_t - 1)$. In summary, we have
\begin{align} \label{eq:rec-original}
(X_{t + 1}, Y_{t + 1}) =
\begin{cases}
(X_t - 1, Y_t + 1) &\text{with probability $\frac{X_t}{X_t + Y_t}$}; \\
(X_t, Y_t - 1) & \text{with probability $\frac{Y_t}{X_t + Y_t}$}.
\end{cases}
\end{align}

Whenever $Y_t = n$, the algorithm exits the while loop and outputs an assignment. However, it will be more convenient for us to study the process with no such stopping condition. In other words, we run the above Markovian process for $t = 1,2, \dots, 2m$, and the probability that our algorithm finds an envy-free assignment is exactly the probability that $\max_{t = 1,2, \dots, 2m} Y_t \geq n$. Note that we choose to run the process for this range of $t$ because $2X_t + Y_t$ decreases by exactly $1$ in each iteration and both $X_t,Y_t$ remain nonnegative throughout by~\eqref{eq:rec-original}, which means that the process is valid for this range and that\footnote{An alternative way to see this is to observe that for $X_t$ to be zero, each item must be assigned once and unassigned once; this happens after exactly $2m$ iterations.} $X_{2m} = Y_{2m} = 0$.

The observation that $2X_t + Y_t$ decreases by $1$ in each iteration is also useful for simplifying our recurrence relation. In particular, it implies that 
\begin{align} \label{eq:invariance}
2X_t + Y_t = 2m - t.
\end{align}
By plugging~\eqref{eq:invariance} into~\eqref{eq:rec-original}, we obtain the following recurrence relation on $X_t$ alone.
\begin{align} \label{eq:rec-simplified}
X_{t + 1} =
\begin{cases}
X_t - 1 &\text{with probability $\frac{X_t}{2m - t - X_t}$}; \\
X_t & \text{with probability $1 - \frac{X_t}{2m - t - X_t}$}.
\end{cases}
\end{align}

\subsection{From Recurrence Relations to Differential Equations and Back}
\label{sec:diff-eq}

One way to quantify the ``average'' change of $X_t$ is via the expected value of $X_{t + 1} - X_t$. In particular, we can use~\eqref{eq:rec-simplified} to calculate
\begin{align*}
\E[X_{t + 1} - X_t \mid X_t = x] = - \frac{x}{2m - t - x} = -\frac{\frac{x}{2m}}{1 - \frac{t}{2m} - \frac{x}{2m}}.
\end{align*}
Let $f(s, z) := -\frac{z}{1 - s - z}$ for $0\leq s,z < 1$. The above equation may be written as 
\begin{align} \label{eq:diff-eq-discrete}
\E[X_{t + 1} - X_t \mid X_t = x] = f\left(\frac{t}{2m}, \frac{x}{2m}\right).
\end{align}
The key idea in the approach of~\citet{KarpSi81} and \citet{Wormald95} is that as $m\rightarrow\infty$, this Markovian process is ``similar'' to the differential equation
\begin{align} \label{eq:diff-eq}
\frac{d z}{d s} = f(s, z) \text{ for } 0\leq s < 1
\end{align}
with the initial condition $z(0) = 1/2$.

This similarity can be formalized. In particular, \citet[Theorem 1]{Wormald95} proved a very general theorem which essentially shows that whenever an equation such as~\eqref{eq:diff-eq-discrete} holds along with some mild technical conditions, we have $\left|\frac{X_t}{2m} - z\left(\frac{t}{2m}\right)\right| = o(1)$ for all indices $t$ with high probability, where $z = z(s)$ is the unique solution to~\eqref{eq:diff-eq}. An application of Wormald's theorem to our setting yields the following:

\begin{lemma} \label{lem:discrete-vs-continuous}
With high probability as $m\rightarrow\infty$,
$|X_t - 2m \cdot z\left(\frac{t}{2m}\right)| = o(m)$
for all $t = 0,1, \dots, 2m - 1$ simultaneously, where $z(s)$ is the solution to~\eqref{eq:diff-eq} in the range $s \in [0, 1)$ and $z(0) = 1/2$.
\end{lemma}

Since the technical constraints in Wormald's result are somewhat cumbersome to state, we defer the full proof of Lemma~\ref{lem:discrete-vs-continuous} to the appendix.

\subsection{Putting Things Together: Proof of Theorem~\ref{thm:ef-assignment}}
\label{sec:proof-ef-assignment}

With all the pieces ready, we now prove our main result of this section.

\begin{proof}[Proof of Theorem~\ref{thm:ef-assignment}]
Let us consider the Markovian process $\{(X_t, Y_t)\}_{0 \leq t \leq 2m}$ defined by~\eqref{eq:rec-original} with $(X_0, Y_0) = (m, 0)$. By Lemma~\ref{lem:discrete-vs-continuous}, with high probability, we have $|X_t - 2m \cdot z\left(\frac{t}{2m}\right)| = o(m)$ for all $t=0,1,\dots,2m-1$, where $z$ is the solution to~\eqref{eq:diff-eq-sol}.
When this holds, we can use~\eqref{eq:invariance} to derive 
\begin{align*}
&\left|Y_t - 2m \cdot z\left(\frac{t}{2m}\right) \cdot \ln\left(\frac{1}{2 z\left(\frac{t}{2m}\right)}\right)\right| \\ 
&\overset{\eqref{eq:invariance}}{=} \left|2m - t - 2X_t - 2m \cdot z\left(\frac{t}{2m}\right) \cdot \ln\left(\frac{1}{2 z\left(\frac{t}{2m}\right)}\right)\right| \\
&\leq \left|2m - t - 4m \cdot z\left(\frac{t}{2m}\right) - 2m \cdot z\left(\frac{t}{2m}\right) \cdot \ln\left(\frac{1}{2 z\left(\frac{t}{2m}\right)}\right)\right| + o(m) \\
&= \left|2m - t - 2m \cdot \left(2 z\left(\frac{t}{2m}\right) - z\left(\frac{t}{2m}\right) \cdot \ln\left(2 z\left(\frac{t}{2m}\right)\right)\right)\right| + o(m) \\
&\overset{\eqref{eq:diff-eq-sol}}{=} \left|2m - t - 2m \cdot \left(1 - \frac{t}{2m}\right)\right| + o(m) \\
&= o(m).
\end{align*}
This means that, with high probability, the following holds (recall that $Y_{2m} = 0$):
\begin{align} \label{eq:y-to-max}
\max_{t = 0,1,\dots, 2m} Y_t &= 2m \cdot \left(\max_{t = 0, 1,\dots, 2m-1}z\left(\frac{t}{2m}\right) \cdot \ln\left(\frac{1}{2 z\left(\frac{t}{2m}\right)}\right)\right) \pm o(m).
\end{align}
Now, observe that $\sup_{s \in [0, 1)} z(s) \ln\left(\frac{1}{2z(s)}\right) = \frac{1}{2e}$, with the maximum achieved at $s^* = 1 - \frac{3}{2e}$ (and $z(s^*) = \frac{1}{2e}$). Clearly, $\max_{t = 0, 1,\dots, 2m-1}z\left(\frac{t}{2m}\right) \cdot \ln\left(\frac{1}{2 z\left(\frac{t}{2m}\right)}\right) \leq \sup_{s \in [0, 1)} z(s) \ln\left(\frac{1}{2z(s)}\right) = \frac{1}{2e}$. On the other hand, we have $\max_{t = 0, 1,\dots, 2m-1}z\left(\frac{t}{2m}\right) \cdot \ln\left(\frac{1}{2 z\left(\frac{t}{2m}\right)}\right) \geq z\left(\frac{\lfloor 2m s^* \rfloor}{2m}\right) \cdot \ln\left(\frac{1}{2 z\left(\frac{\lfloor 2m s^* \rfloor}{2m}\right)}\right)$. Moreover, one can check that $z(s)$ is continuous at $s = s^*$, which means that $\lim_{m \to \infty} z\left(\frac{\lfloor 2m s^* \rfloor}{2m}\right) \cdot \ln\left(\frac{1}{2 z\left(\frac{\lfloor 2m s^* \rfloor}{2m}\right)}\right) = z(s^*) \cdot \ln\left(\frac{1}{2z(s^*)}\right) = \frac{1}{2e}$. Combining these lower and upper bounds, we have
\begin{align*}
\max_{t = 0, 1,\dots, 2m-1}z\left(\frac{t}{2m}\right) \cdot \ln\left(\frac{1}{2 z\left(\frac{t}{2m}\right)}\right) = \frac{1}{2e} + o(1)
\end{align*}
with high probability, where the term $o(1)$ converges to zero as $m \to \infty$. Plugging this back into~\eqref{eq:y-to-max}, we get
\begin{align*}
\max_{t = 0,1,\dots, 2m} Y_t = m/e \pm o(m).
\end{align*}

Finally, recall that the probability that \textsc{GreedyAssignment} finds an envy-free assignment is exactly the probability that $\max_{t=0,1,\dots,2m} Y_t \geq n$. Thus, if $m/n \geq e + \varepsilon$ for some $\varepsilon > 0$, then with high probability we have 
\begin{align*}
\max_{t=0,1,\dots,2m} Y_t \geq \frac{m}{e} - o(m) = \frac{m}{e + \varepsilon} + \frac{m \epsilon}{e(e + \epsilon)} - o(m) \geq n + m \cdot \left(\frac{\epsilon}{e(e + \epsilon)} - o(1)\right),
\end{align*}
which is at least $n$ for any sufficiently large $m$. This implies that \textsc{GreedyAssignment} finds an envy-free assignment with high probability in this case. On the other hand, if $m/n \leq e - \varepsilon$ for some $\varepsilon > 0$, then, using a similar argument, we can conclude that \textsc{GreedyAssignment} outputs NULL with high probability, in which case Lemma~\ref{lem:greedy-correctness} implies that no envy-free assignment exists.
\end{proof}

\section{Conclusion and Future Work}

In this paper, we have studied the asymptotic existence of fair allocations and settled several open questions from previous work.
In addition to the tight bounds themselves, our work also sheds light on the fairness guarantees provided by different algorithms in the probabilistic setting.
Specifically, our results serve as a strong argument for using the classical round-robin algorithm when allocating indivisible items: not only is the algorithm simple and its output always envy-free up to one item (EF1), but the produced allocation is likely to be fully envy-free as well as proportional provided that the number of items is sufficiently larger than the number of agents.
We also show that an EFX allocation exists with high probability for any relation between the numbers of agents and items, further confirming the worst-case existence of such allocations as a tantalizing open question.\footnote{Recently, \citet{ChaudhuryGaMe20} showed that an EFX allocation always exists when there are three agents.}

An interesting avenue that remains after this work is to investigate the asymptotic behavior of fair allocations that satisfy additional properties.
In fact, some desirable properties are already implied by previous results---for example, \citet{DickersonGoKa14} showed that a welfare-maximizing allocation is envy-free with high probability assuming that $m=\Omega(n\log n)$, while \citet{ManurangsiSu19} proved that if $m$ is a multiple of $n$, there exists an algorithm that likely computes an allocation which is both envy-free and \emph{balanced} (i.e., gives every agent the same number of items) as long as $m\geq 2n$.
In a similar vein, one could examine common fair division algorithms and solutions such as the envy-cycle elimination algorithm \citep{LiptonMaMo04}, the maximum Nash welfare solution \citep{CaragiannisKuMo16}, or the leximin solution \cite{BogomolnaiaMo04,KurokawaPrSh15} through the probabilistic lens.
%More generally, combining fairness notions with properties like economic efficiency \citep{CaragiannisKuMo16,BarmanKrVa18} and connectivity \citep{BiloCaFl19,Suksompong19} is an exciting direction for future research.

Our asymptotic approach can also be applied beyond the canonical resource allocation setting in which the resource is allocated to individual agents who have equal entitlements.
For instance, many practical situations entail dividing items among \emph{groups} of agents---the agents in each group share the same set of items but may have different opinions on them \citep{Suksompong18,SuksompongPhD}.
In this generalized setting, \citet{ManurangsiSu17} studied the asymptotic existence of envy-free allocations and left open a logarithmic gap between existence and non-existence.
Likewise, a number of division problems involve agents who have different entitlements to the resource \citep{BabaioffNiTa19,FarhadiGhHa19}.
The definition of envy-freeness can be naturally extended to capture such scenarios, and \citet{ChakrabortyIgSu20} demonstrated through experiments that weighted envy-free allocations are
usually harder to find than their unweighted counterparts. 
Providing a formal explanation for this phenomenon using probabilistic tools is an intriguing direction for future research.

% Bibliography
\bibliographystyle{ACM-Reference-Format}
\bibliography{main}

% Appendix
\appendix

\section{Omitted Proofs}

\subsection*{Proof of Lemma~\ref{lem:rr-simplified-process}}
%\label{app:rr-simplified-process}

To prove Lemma~\ref{lem:rr-simplified-process}, we will need the lemma below, which states that the following two processes result in the same distribution: (1) sample $\ell$ i.i.d. random variables from a distribution $\cD$, and let $Y^{\max}$ be the maximum value and $Z_1, \dots, Z_{\ell - 1}$ be the remaining values; and (2) sample the first random variable $W^{\max}$ from $\cD^{\max(m)}$, and sample $W_1, \dots, W_{\ell - 1}$ independently from the distribution $\cD$ conditioned on the value being at most $W^{\max}$ (i.e., $\cD_{\leq W^{\max}}$).

\begin{lemma} \label{lem:conditioned-max}
Let $\ell$ be a positive integer and let $\cD$ be any non-atomic probability distribution (over real numbers). Consider the following two processes:
\begin{itemize}
\item Sample $Y_1, \dots, Y_\ell$ i.i.d. at random from $\cD$. Let $q^* = \argmax_{q \in [\ell]}\{Y_q\}$ (tie broken arbitrarily). Then, let $Y^{\max} = Y_{q^*}$ and,  for each $t \in [\ell - 1]$, let 
\begin{align*}
Z_t =
\begin{cases}
Y_{t} &\text{ if } 1 \leq t < q^*; \\
Y_{t + 1} &\text{ if } q^* \leq t \leq \ell-1.
\end{cases}
\end{align*}
\item Sample $W^{\max} \sim \cD^{\max(\ell)}$. Then, sample $W_1, \dots, W_{\ell - 1}$ i.i.d. at random from  $\cD_{\leq W^{\max}}$.
\end{itemize}
Then, $(Y^{\max}, Z_1, \dots, Z_{\ell - 1})$ and $(W^{\max}, W_1, \dots, W_{\ell - 1})$  are identically distributed.
\end{lemma}

\begin{proof}
For any $y^{\max}, z_1, \dots, z_{\ell - 1} \in \R$, we may write $\Pr[Y^{\max} \leq y^{\max}, Z_1 \leq z_1, \dots, Z_{\ell-1} \leq z_{\ell-1}]$ as
\begin{align*}
\Pr[Y^{\max} \leq y^{\max}, Z_1 \leq z_1, \dots, Z_{\ell-1} \leq z_{\ell-1}] &= \sum_{q \in [\ell]} \Pr[Y^{\max} \leq y^{\max}, Z_1 \leq z_1, \dots, Z_{\ell - 1} \leq z_{\ell-1}, q^{*} = q] \\
&= \ell \cdot \Pr[Y^{\max} \leq y^{\max}, Z_1 \leq z_1, \dots, Z_{\ell - 1} \leq z_{\ell-1}, q^{*} = \ell],
\end{align*}
where the equality follows from symmetry over the random variables $Y_1, \dots, Y_\ell$ and the fact that since $\cD$ is non-atomic, the probability of a tie (for $q^*$) is zero. We may rewrite the above expression further as
\begin{align}
&\Pr[Y^{\max} \leq y^{\max}, Z_1 \leq z_1, \dots, Z_{\ell - 1} \leq z_{\ell-1}] \nonumber \\
&= \ell \cdot \Pr[(Y_\ell \geq Y_1, \dots, Y_{\ell - 1}), Y_\ell \leq y^{\max}, Y_1 \leq z_1, \dots, Y_{\ell - 1} \leq z_{\ell - 1}] \nonumber \\
&= \ell \cdot \int_{-\infty}^{y^{\max}} \int_{-\infty}^{\min\{a^{\max}, z_1\}} \cdots \int_{-\infty}^{\min\{a^{\max}, z_{\ell - 1}\}}  f_{\cD}(a_1) \cdots f_{\cD}(a_{\ell - 1}) f_{\cD}(a^{\max}) \quad  d a_{\ell - 1} \cdots d a_1 d a^{\max} \nonumber \\
&= \ell \cdot \int^{y^{\max}}_{\infty} \left(\int_{-\infty}^{\min\{a^{\max}, z_1\}} f_{\cD}(a_1) da_1\right) \cdots \left(\int_{-\infty}^{\min\{a^{\max}, z_{\ell - 1}\}} f_{\cD}(a_{\ell - 1}) da_{\ell - 1}\right) f_{\cD}(a^{\max}) da^{\max} \nonumber \\
&= \ell \cdot \int^{y^{\max}}_{\infty}  F_{\cD}(\min\{a^{\max}, z_1\}) \cdots F_{\cD}(\min\{a^{\max}, z_{\ell - 1}\}) f_{\cD}(a^{\max}) da^{\max} \label{eq:expanded-CDF-first}
\end{align}
We will next expand a similar term for $(W^{\max}, W_1, \dots, W_{\ell - 1})$. Before we do so, let us recall that, due to the definition of $\cD^{\max(\ell)}$, we have $F_{\cD^{\max(\ell)}}(z) = F_{\cD}(z)^\ell$ for any $z \in \R$. Taking the derivative of both sides, we have
\begin{align} \label{eq:max-pdf}
f_{\cD^{\max(\ell)}}(z) = \ell \cdot F_{\cD}(z)^{\ell - 1} \cdot f_{\cD}(z).
\end{align}
Now, for any $w^{\max}, w_1, \dots, w_{\ell - 1} \in \R$, we can expand $\Pr[W^{\max} \leq w^{\max}, W_1 \leq w_1, \dots, W_{\ell - 1} \leq w_{\ell - 1}]$ as
\begin{align}
&\Pr[W^{\max} \leq w^{\max}, W_1 \leq w_1, \dots, W_{\ell - 1} \leq w_{\ell - 1}] \nonumber \\
&= \int^{w^{\max}}_{-\infty} \left(F_{\cD_{\leq a^{\max}}}(w_1) \cdots F_{\cD_{\leq a^{\max}}}(w_{\ell - 1})\right) f_{\cD^{\max(\ell)}}(a^{\max}) da^{\max} \nonumber \\
&\overset{\eqref{eq:range-conditioned-CDF}}{=} \int^{w^{\max}}_{-\infty} \left(\frac{F_{\cD}(\min\{a^{\max}, w_{1}\})}{F_{\cD}(a^{\max})} \cdots \frac{F_{\cD}(\min\{a^{\max}, w_{\ell - 1}\})}{F_{\cD}(a^{\max})}\right) f_{\cD^{\max(\ell)}}(a^{\max}) da^{\max} \nonumber \\
&\overset{\eqref{eq:max-pdf}}{=} \ell \cdot \int^{w^{\max}}_{\infty}  F_{\cD}(\min\{a^{\max}, w_1\}) \cdots F_{\cD}(\min\{a^{\max}, w_{\ell - 1}\}) f_{\cD}(a^{\max}) da^{\max}. \label{eq:expanded-CDF-second}
\end{align}
From~\eqref{eq:expanded-CDF-first} and~\eqref{eq:expanded-CDF-second}, we conclude that $(Y^{\max}, Z_1, \dots, Z_\ell)$ and $(W^{\max}, W_1, \dots, W_{\ell - 1})$ are identically distributed.
\end{proof}

Now, the main idea for the proof of Lemma~\ref{lem:rr-simplified-process} is to apply Lemma~\ref{lem:conditioned-max} repeatedly to gradually transfer the process from the original round-robin process to the one described in Lemma~\ref{lem:rr-simplified-process}.

\begin{proof}[Proof of Lemma~\ref{lem:rr-simplified-process}]
Recall that the round-robin process can be written as follows:
\begin{enumerate}
\item For every $i \in [n], j \in [m]$, let $u_i(j) \sim \cD_{\leq X^i_0}$.
\item For $t = 1, \dots, \lceil m/n \rceil$:
\begin{enumerate}
\item For $i = 1, \dots, \min\{n, m - (t - 1)n\}$:
\begin{enumerate}
\item Let $j^*$ be the index of a remaining item that maximizes $u_i(j^*)$. \label{step:rr-pick}
\item Remove $j^*$ from the set of available items.
\item Set $X^i_t = u_i(j^*)$ and, for every $i' \in [n] \setminus \{i\}$, set $X^{i', i}_t = u_{i'}(j^*)$.
\end{enumerate}
\end{enumerate}
\end{enumerate}

Consider the first item $j^*_1$ picked by the first agent (in Step~\ref{step:rr-pick} when $t = 1$ and $i = 1$). For notational convenience, let us assume without loss of generality that $j^*_1 = m$. 
From Lemma~\ref{lem:conditioned-max} with $\ell = m$, we have that $X^1_1$ is distributed as $\cD^{\max(m)} = \cD^{\max(m)}_{\leq X_0^1}$ (recall that $X_0^1=1$), and that, for the remaining items $j \in [m - 1]$, the utilities $u_1(j)$ are distributed i.i.d. as $\cD_{\leq X_1^1}$. Moreover, since agent~1 does not consider other agents' utilities at all when picking $j^*_1$, we also have that $X^{2, 1}_1, \dots, X^{n, 1}_1$ are distributed i.i.d. as $\cD$ and, for the remaining items $j \in [m - 1]$, $u_2(j), \dots, u_n(j)$ are distributed i.i.d. as $\cD$. From the observations so far, the round-robin process is equivalent to the following process: 
\begin{enumerate}
\item Sample utilities of the first item selected by the first agent: \label{step:rr-first-pick}
\begin{itemize}
\item Sample $X^1_1 \sim \cD^{\max(m)}_{\leq X^1_0}$.
\item For every $1 < i \leq n$, sample $X^{i, 1}_t \sim \cD_{\leq X^{i}_0}$.
\end{itemize}
\item Sample utilities of the remaining items:
\begin{itemize}
\item For every $j \in [m - 1]$, sample $u_1(j) \sim \cD_{\leq X^1_1}$.
\item For every $1 < i \leq n$ and $j \in [m - 1]$, sample $u_i(j) \sim \cD_{\leq X^i_0}$.
\end{itemize}
\item For $t = 1, \dots, \lceil m/n \rceil$:
\begin{enumerate}
\item For $i = \max\{1, 2 \cdot \bone[t = 1]\}, \dots, \min\{n, m - (t - 1)n\}$: \label{step:rr-loop-after-first-transformation}
\begin{enumerate}
\item Let $j^*$ be the index of a remaining item that maximizes $u_i(j^*)$.
\item Remove $j^*$ from the set of available items.
\item Set $X^i_t = u_i(j^*)$ and, for every $i' \in [n] \setminus \{i\}$, set $X^{i', i}_t = u_{i'}(j^*)$.
\end{enumerate}
\end{enumerate}
\end{enumerate}
(We use $\bone[E]$ to denote the indicator random variable for event $E$. Note that in Step~\ref{step:rr-loop-after-first-transformation}, the term $2 \cdot \bone[t = 1]$ is there so that the first agent does not get to pick in the first round, since this pick was already taken care of in Step~\ref{step:rr-first-pick}.)

Similarly to the arguments above, we may now consider the first item $j^*_2$ picked by the second agent (in Step~\ref{step:rr-loop-after-first-transformation} when $t = 1$ and $i = 2$) and assume without loss of generality that $j^*_2 = m - 1$. From Lemma~\ref{lem:conditioned-max}, $X^2_1$ is distributed as $\cD^{\max(m-1)} = \cD^{\max(m-1)}_{\leq X_0^2}$ (recall that $X_0^2=1$), and, for the remaining items $j \in [m - 2]$, $u_2(j)$ is distributed i.i.d. as $\cD_{\leq X_1^2}$. Moreover, since agent 2 does not consider other agent's utilities at all when picking $j^*_2$, we also have that $X^{1, 2}_1, X^{3, 2}_1 , \dots, X^{n, 2}_1$ are independently distributed as $\cD_{\leq X^1_1}, \cD, \dots, \cD$ respectively, and, for the remaining items $j \in [m - 2]$, $u_1(j), u_3(j), \dots, u_n(j)$ are independently distributed as $\cD_{\leq X^1_1}, \cD, \dots, \cD$ respectively. As a result, the process above is in turn equivalent to the following process.

\begin{enumerate}
\item Sample utilities of the first item selected by the first agent:
\begin{itemize}
\item Sample $X^1_1 \sim \cD^{\max(m)}_{\leq X^1_0}$.
\item For every $1 < i \leq n$, sample $X^{i, 1}_1 \sim \cD_{\leq X^{i}_0}$.
\end{itemize}
\item Sample utilities of the first item selected by the second agent:
\begin{itemize}
\item Sample $X^2_1 \sim \cD^{\max(m-1)}_{\leq X^2_0}$.
\item Sample $X^{1, 2}_1 \sim \cD_{\leq X^1_1}$.
\item For every $2 < i \leq n$, sample $X^{i, 2}_1 \sim \cD_{\leq X^{i}_0}$.
\end{itemize}
\item Sample utilities of the remaining items:
\begin{itemize}
\item For every $j \in [m - 2]$, let $u_1(j) \sim \cD_{\leq X^1_1}$ and $u_2(j) \sim \cD_{\leq X^2_1}$.
\item For every $2 < i \leq n$ and $j \in [m - 2]$, let $u_i(j) \sim \cD_{\leq X^i_0}$.
\end{itemize}
\item For $t = 1, \dots, \lceil m/n \rceil$:
\begin{enumerate}
\item For $i = \max\{1, 3 \cdot \bone[t = 1]\}, \dots, \min\{n, m - (t - 1)n\}$: 
\begin{enumerate}
\item Let $j^*$ be the index of a remaining item that maximizes $u_i(j^*)$.
\item Remove $j^*$ from the set of available items.
\item Set $X^i_t = u_i(j^*)$ and, for every $i' \in [n] \setminus \{i\}$, set $X^{i', i}_t = u_{i'}(j^*)$.
\end{enumerate}
\end{enumerate}
\end{enumerate}

By repeatedly applying this argument $m - 2$ additional times, we will arrive at the process stated in Lemma~\ref{lem:rr-simplified-process}, and the proof is complete.
\end{proof}

\subsection*{Proof of Lemma~\ref{lem:discrete-vs-continuous}}
%\label{app:discrete-vs-cont}

We explain how Theorem 1 of~\citet{Wormald95} implies our Lemma~\ref{lem:discrete-vs-continuous}. To do so, we first restate Wormald's theorem for the special case of a single sequence of random variables:

\begin{theorem}[\citet{Wormald95}] \label{thm:wormald}
Let $(X_t)_{0 \leq t \leq T}$ be a Markovian random process such that $X_0/T = x^* \in \R$, $|X_{t} - X_{t + 1}| \leq 1$ for all $t = 0 ,1,\dots, T - 1$, and there exists a function $f: \R^2 \to \R$ such that $$\E[X_{t + 1} - X_t \mid X_t = x] = f(t/T, x/T)$$ for all $t \in \{0,1, \dots, T - 1\}$.

Suppose further that there exists a bounded connected open set $D \subseteq \R^2$ such that $f$ is continuous and satisfies the Lipschitz condition\footnote{That is, there exists a constant $L > 0$ such that $f(s, z) - f(s', z') \leq L \cdot\max\{|s - s'|, |z - z'|\}$ for all $(s, z), (s', z') \in D$. The constant $L$ is said to be a \emph{Lipschitz constant} of $f$.} on $D$, and $(0, x^*) \in D$.

Then, the differential equation $\frac{dz}{ds} = f(s, z)$ with the initial condition $z(0) = x^*$ has a unique solution $z(s)$ on $D$. Furthermore, with high probability as $T\rightarrow\infty$, the following holds: for every $t$ such that $(t/T, z(t/T)) \in D$, we have $$|X_t - T \cdot z(t/T)| = o(T).$$
\end{theorem}

At first glance, it may seem that the above theorem immediately implies our Lemma~\ref{lem:discrete-vs-continuous}. Nonetheless, there is in fact a slightly subtle point, because the Lipschitz constant of our function $f$ is not bounded as $s \to 1$. However, this is a common issue and was also faced by Wormald in his original paper~\cite{Wormald95}. Wormald handled this by using the concentration inequality only for $s \leq 1 - \varepsilon$ and then use the fact that $|X_t - X_{t + 1}| \leq 1$ to deal with the rest of the range (i.e., $t \geq (1 - \varepsilon)T$). A similar approach works for us here, as formalized below.

\begin{proof}[Proof of Lemma~\ref{lem:discrete-vs-continuous}]
First, note that our differential equation~\eqref{eq:diff-eq} can be easily solved via standard methods, and its solution is the unique $z = z(s)$ that satisfies
\begin{align} \label{eq:diff-eq-sol}
2z - z \ln(2z) = 1 - s.
\end{align}

Let $0 < \varepsilon < 1$ be any constant. We will argue that, with high probability as $m\rightarrow\infty$, we have $|X_t - 2m \cdot z\left(\frac{t}{2m}\right)| \leq \varepsilon m$ for all $t = 0, 1, \dots, 2m - 1$.

Consider the set $D = \{(s, z) \mid -0.1 < s < 1, -0.1 < s + z < 1 - 0.01\varepsilon\}$. For any $(s, z), (s', z') \in D$, we have
\begin{align*}
f(s, z) - f(s', z') &= -\frac{z}{1 - s - z} + \frac{z'}{1 - s' - z'} \\
&= \frac{z' - z + zs' - z's}{(1 - s - z)(1 - s' - z')} \\
&= \frac{(z' - z) + (zs' - z's') + (z's' - z's)}{(1 - s - z)(1 - s' - z')} \\
&\leq \frac{|z' - z| + |s'| \cdot |z - z'| + |z'| \cdot |s' - s|}{(1 - s - z)(1 - s' - z')} \\
(\text{From } |s'| < 1 \text{ and } |z'| < 1.1) &\leq \frac{2|z' - z| + 1.1|s' - s|}{(1 - s - z)(1 - s' - z')} \\
(\text{From } s + z, s' + z' < 1 - 0.01\varepsilon) &\leq \frac{31000}{\varepsilon^2} \cdot \max\{|z' - z|, |s' - s|\},
\end{align*}
which means that $f$ satisfies the Lipschitz condition on $D$ (with Lipschitz constant $31000/\varepsilon^2$).

As a result, by applying Theorem~\ref{thm:wormald} with $T = 2m$, the following holds with high probability: for all $t$ such that $\frac{t}{2m} + z\left(\frac{t}{2m}\right) < 1 - 0.01\varepsilon$, we have
\begin{align} \label{eq:lipschitz-range}
\left|X_t - 2m \cdot z\left(\frac{t}{2m}\right)\right| \leq o(m).
\end{align}
Let $s^* \in [0, 1)$ be such that, in our equation~\eqref{eq:diff-eq-sol}, $z(s^*) = 0.1\epsilon$.
Note that there exists a unique such $s^*$ because $z(s)$ is decreasing and continuous for $s \in [0, 1)$, $z(0) = 1/2$, and $\lim_{s \to 1^{-}} z(s) = 0$.
Moreover, we have $z(s) > 0.1\varepsilon$ for $s < s^*$.
Let $t^* = \lceil s^* T \rceil$.

Notice that we have $s + z(s) = 1 - z(s) + z(s)\ln(2z(s)) < 1 - 0.1z(s)$ for all $s\in [0,1)$. In other words, $\frac{t}{2m} + z(\frac{t}{2m}) < 1 - 0.1z\left(\frac{t}{2m}\right) \leq 1 - 0.01\varepsilon$ for any integer $t < t^*$, which means that \eqref{eq:lipschitz-range} is satisfied for all $t<t^*$ with high probability.  

On the other hand, to see that~\eqref{eq:lipschitz-range} is also likely to hold for $t \geq t^*$, first observe that the sequence $(X_t)_{0 \leq t \leq T}$ is non-increasing. Hence, for such $t$ we have
\begin{align*}
X_t \leq X_{t^* - 1} \overset{\eqref{eq:lipschitz-range}}{\leq} 2m \cdot z\left(\frac{t^* - 1}{2m}\right) + o(m). 
\end{align*}
Now, since $z(s)$ is continuous at $s^*$, it must be the case that $\frac{t^* - 1}{2m}$ converges to $s^*$ as $m$ grows.
This means that $\left|z\left(\frac{t^* - 1}{2m}\right) - z(s^*)\right| = o(1)$, where the $o(1)$ term converges to zero as $m \to \infty$. Plugging this into the inequality above, we get
\begin{align} \label{eq:large-t-bound}
X_t &\leq 2m \cdot z(s^*) + o(m) = 0.2 \epsilon m + o(m),
\end{align}
where the equality follows from our choice of $s^*$. 
Thus, we have
\begin{align*}
- \epsilon m &\leq -0.2 \epsilon m \\
(\text{From our choice of } s^*) &= -2m \cdot z(s^*) \\
(\text{Since } z(s) \text{ is decreasing and } t \geq t^*) &\leq -2m \cdot z\left(\frac{t}{2m}\right) \\
&\leq X_t -2m \cdot z\left(\frac{t}{2m}\right) \\
&\leq X_t \\
(\text{From }\eqref{eq:large-t-bound}) &\leq 0.2 \varepsilon m + o(m),
\end{align*}
which is at most $\varepsilon m$ for any sufficiently large $m$. This implies that with high probability, $|X_t - 2m \cdot z\left(\frac{t}{2m}\right)| \leq \varepsilon m$ for all $t \geq t^*$. In conclusion, we have $|X_t - 2m \cdot z\left(\frac{t}{2m}\right)| \leq \varepsilon m$ for all $t \in \{0,1, \dots, 2m - 1\}$ with high probability, as desired.
\end{proof}

\end{document}